\documentclass[preprint,12pt,number]{elsarticle} 

\usepackage{amsmath,amssymb,amsfonts}
\usepackage{amsthm}

\usepackage{xcolor}
\usepackage{float}
\usepackage{graphicx}
\usepackage{mathrsfs}
\usepackage{hyperref}

\usepackage{listings}

\usepackage{tikz}
\usetikzlibrary{positioning, calc, arrows.meta,automata,positioning}
\usepackage{subcaption}

\usepackage{cleveref}

\newtheorem{theorem}{Theorem}[section]
\newtheorem{lemma}[theorem]{Lemma}
\theoremstyle{remark}

\DeclareUnicodeCharacter{0301}{\'{e}}
\usepackage{accents}

\newcommand{\D}{\mathrm{d}}

\newcommand{\tc}{{\widetilde{c}}}

\newcommand{\tK}{{\widetilde{K}}}

\newcommand{\ts}{{\widetilde{s}}}

\newcommand{\ty}{{\widetilde{y}}}

\newcommand{\talpha}{{\widetilde{\alpha}}}

\newcommand{\bD}{{\boldsymbol{D}}}

\newcommand{\bg}{{\boldsymbol{g}}}

\newcommand{\bk}{{\boldsymbol{k}}}

\newcommand{\bt}{{\boldsymbol{t}}}

\newcommand{\bu}{{\boldsymbol{u}}}

\newcommand{\bv}{{\boldsymbol{v}}}

\newcommand{\bx}{{\boldsymbol{x}}}
\newcommand{\bX}{{\boldsymbol{X}}}
\newcommand{\by}{{\boldsymbol{y}}}

\newcommand{\bz}{{\boldsymbol{z}}}

\newcommand{\balpha}{{\boldsymbol{\alpha}}}

\newcommand{\bDelta}{{\boldsymbol{\Delta}}}

\newcommand{\bkappa}{{\boldsymbol{\kappa}}}
\newcommand{\blambda}{{\boldsymbol{\lambda}}}

\newcommand{\bomega}{{\boldsymbol{\omega}}}

\newcommand{\tbw}{{\widetilde{\boldsymbol{w}}}}

\newcommand{\tby}{{\widetilde{\boldsymbol{y}}}}

\newcommand{\hf}{{\widehat{f}}}

\newcommand{\hmu}{{\widehat{\mu}}}

\newcommand{\hsigma}{{\widehat{\sigma}}}

\newcommand{\fu}{{\mathfrak{u}}}

\newcommand{\mA}{{\mathsf{A}}}

\newcommand{\mC}{{\mathsf{C}}}

\newcommand{\mF}{{\mathsf{F}}}

\newcommand{\mH}{{\mathsf{H}}}
\newcommand{\mi}{{\mathsf{i}}}
\newcommand{\mI}{{\mathsf{I}}}

\newcommand{\mk}{{\mathsf{k}}}
\newcommand{\mK}{{\mathsf{K}}}

\newcommand{\mS}{{\mathsf{S}}}
\newcommand{\mT}{{\mathsf{T}}}

\newcommand{\mv}{{\mathsf{v}}}
\newcommand{\mV}{{\mathsf{V}}}

\newcommand{\mx}{{\mathsf{x}}}

\newcommand{\my}{{\mathsf{y}}}

\newcommand{\mz}{{\mathsf{z}}}

\newcommand{\mDelta}{{\mathsf{\Delta}}}

\newcommand{\mLambda}{{\mathsf{\Lambda}}}

\newcommand{\mPi}{{\mathsf{\Pi}}}

\newcommand{\bmi}{{\boldsymbol{\mathsf{i}}}}

\newcommand{\bmz}{{\boldsymbol{\mathsf{z}}}}

\newcommand{\tmC}{{\widetilde{\mathsf{C}}}}

\newcommand{\calO}{{\mathcal{O}}}

\newcommand{\calP}{{\mathcal{P}}}

\newcommand{\calU}{{\mathcal{U}}}

\newcommand{\bbB}{{\mathbb{B}}}

\newcommand{\bbC}{{\mathbb{C}}}

\newcommand{\bbE}{{\mathbb{E}}}

\newcommand{\bbN}{{\mathbb{N}}}

\newcommand{\bbR}{{\mathbb{R}}}

\newcommand{\llvert}{\left\lvert}
\newcommand{\rrvert}{\right\rvert}

\newcommand{\bzero}{{\boldsymbol{0}}}

\newcommand{\simiid}{\overset{\mathrm{IID}}{\sim}}

\newcommand{\diag}{\mathrm{diag}}
\newcommand{\wal}{\mathrm{wal}}

\usepackage{listings}
\definecolor{darkgreen}{rgb}{0,0.6,0}
\lstdefinestyle{Python}{
    showstringspaces=false,
    language        = Python,
    basicstyle      = \small\ttfamily,
    morekeywords = {as},
    keywordstyle    = \color{blue},
    stringstyle     = \color{purple},
    commentstyle    = \color{darkgreen}\ttfamily,
    breaklines = true,
	postbreak=\text{$\hookrightarrow$\space},
	alsoletter = {>,.} ,
    morekeywords = [2]{>>>,...},
    keywordstyle = [2]\color{cyan}\bfseries}

\newcommand{\thisDelta}{0.227}
\newcommand{\odiv}{\mathrel{\ooalign{$\bigcirc$\cr\hidewidth$\div$\hidewidth}}}

\begin{document}

\begin{frontmatter}

\title{QMCPy: A Python package for randomized low-discrepancy sequences, quasi-Monte Carlo, and fast kernel methods}

\author[inst1,inst2]{Aleksei G.~Sorokin}
\ead{sorokin@uchicago.edu}

\address[inst1]{Department of Applied Mathematics, Illinois Institute of Technology, Chicago, USA}
\address[inst2]{Department of Statistics, University of Chicago, Chicago, USA}

\begin{abstract}
    Low-discrepancy (LD) sequences are widely used as efficient experimental designs for high-dimensional numerical integration and function approximation. This article presents \texttt{QMCPy}, an open-source Python library that provides a unified framework for randomized LD sequences, quasi-Monte Carlo (QMC) methods, and fast kernel-based computations. We systematically describe the supported rank-$1$ lattices, digital nets (including higher-order constructions), and Halton point sets, together with randomization techniques such as random shifts, linear matrix scrambling (LMS), nested uniform scrambling (NUS), digital shifts, and digital permutations. We emphasize practical implementation issues such as extensible sequence generation, Gray code ordering, and efficient digital operations. Beyond integration, \texttt{QMCPy} supports fast kernel methods in reproducing kernel Hilbert spaces (RKHSs) by pairing LD point sets with shift-invariant (SI) and digitally shift-invariant (DSI) kernels, including higher-order variants, which yields structured Gram matrices. In particular, the resulting Gram matrices have circulant or recursive symmetric block Toeplitz (RSBT) structure, allowing the costs of matrix-vector products and linear solves to be reduced from $\mathcal{O}(n^2)- \mathcal{O}(n^3)$ to $\mathcal{O}(n \log n)$ by using fast Fourier transforms (FFTs) and fast Walsh--Hadamard transforms (FWHTs). We derive a new computable form of an order-$4$ DSI kernel, develop efficient eigenvalue and transform-update algorithms, and present numerical experiments that demonstrate the accuracy, convergence rates, and computational efficiency of the implemented methods across a range of test integrands and dimensions. These capabilities in \texttt{QMCPy} provide a practical, reproducible platform for applying randomized QMC and kernel-based techniques in computational science and engineering.
\end{abstract}


\end{frontmatter}

\section{Introduction} \label{sec:introduction}

\subsection{Background} \label{sec:background}

Low-discrepancy (LD) sequences, also called quasi-random sequences, are designed to judiciously explore the unit cube in high dimensions. They were first developed as replacements for independent samples in Monte Carlo simulation for high-dimensional numerical integration. The resulting quasi-Monte Carlo (QMC) methods have been shown to outperform independent Monte Carlo methods for a broad class of sufficiently regular integrands; see, for example, \citep{hlawka1961funktionen,niederreiter.qmc_book,dick.digital_nets_sequences_book,dick2022lattice,sloan1994lattice,dick.high_dim_integration_qmc_way,sobol1967distribution,halton1960efficiency,sloan1998quasi}. This advantage has been demonstrated in a variety of scientific and engineering applications, including financial modeling \citep{joy1996quasi,lai1998applications,l2004quasi,l2009quasi,xu2015high,giles.mlqmc_path_simulation}, elliptic partial differential equations (PDEs) with random coefficients \citep{graham2011quasi,kuo2012quasi,kuo2015multi,graham2015quasi,kuo.application_qmc_elliptic_pde,robbe.multi_index_qmc}, and graphics rendering via ray tracing \citep{jensen2003monte,raab2006unbiased,waechter2011quasi}, among others. Randomizations of LD sequences provide further advantages, such as enabling practical error estimation, avoiding boundary observations, and reducing the risk of pathological alignments with adverse integrands \citep{owen1995randomly,l2016randomized,lecuyer.RQMC_CLT_bootstrap_comparison,owen.variance_alternative_scrambles_digital_net,MATOUSEK1998527,tezuka2002randomization}.

Randomized LD sequences also play an important role as experimental designs in reproducing kernel Hilbert spaces (RKHSs). For suitable pairings of LD point sets and RKHS kernels, the corresponding Gram matrices exhibit structure that can be exploited by fast transforms \citep{zeng.spline_lattice_digital_net,zeng.spline_lattice_error_analysis}. For example, pairing a rank-$1$ lattice in linear order with a shift-invariant (SI) kernel yields a circulant Gram matrix diagonalizable by fast Fourier transforms (FFTs), while pairing a base-$2$ digital net with an appropriate digitally shift-invariant (DSI) kernel yields a recursive symmetric block Toeplitz (RSBT) structured Gram matrix diagonalizable by fast Walsh--Hadamard transforms (FWHTs). Such constructions have been used to accelerate kernel interpolation for PDEs with random coefficients \citep{kaarnioja.kernel_interpolants_lattice_rkhs,kaarnioja.kernel_interpolants_lattice_rkhs_serendipitous,sorokin.gp4darcy}, Bayesian cubature \citep{rathinavel.bayesian_QMC_lattice,rathinavel.bayesian_QMC_sobol,rathinavel.bayesian_QMC_thesis}, and worst-case error (WCE) computations \citep{hickernell.generalized_discrepancy_quadrature_error_bound,hickernell1998lattice}. LD point sets have also been investigated as training designs in scientific machine learning \citep{longo2021higher,chen2021quasi,keller2025regularity}, although we do not pursue this direction here.

Despite the maturity of the underlying theory, software support for randomized LD sequences and their use in both QMC and fast kernel methods remains fragmented. Existing implementations are often distributed across disjoint libraries or limited in their support for higher-order constructions, randomization schemes, and fast kernel methods. Our prior works \citep{choi.QMC_software,choi.challenges_great_qmc_software,hickernell.qmc_what_why_how,sorokin.MC_vector_functions_integrals,sorokin.qmcpy_joss} reviewed this software landscape and introduced \texttt{QMCPy}, an open-source Python framework for accessible QMC routines. Those works focused on the basic design and initial capabilities of \texttt{QMCPy} and on the broader context of QMC software.

The present article builds on that foundation and has a different emphasis: we provide a coherent mathematical and algorithmic description of the expanded suite of randomized LD sequences now supported in \texttt{QMCPy}, and we develop and implement fast kernel methods based on structured SI and DSI kernels paired with LD point sets. In doing so, we aim to bridge the gap between the rich theoretical literature on QMC and kernel methods and a practical, reproducible Python implementation that can be readily used in computational science and engineering.

\subsection{\texttt{QMCPy} architecture and implementation}

The features discussed in this article are available in \texttt{QMCPy} (\url{https://qmcsoftware.github.io/QMCSoftware/}) \citep{choi.QMC_software,choi.challenges_great_qmc_software,hickernell.qmc_what_why_how,sorokin.MC_vector_functions_integrals,sorokin.qmcpy_joss} version 2.0 or later, which may be installed from the Python package index (PyPI) using the command:

\begin{lstlisting}[style=Python]
pip install -U qmcpy
\end{lstlisting}

\noindent To follow along with the code snippets in this paper, we assume both \texttt{QMCPy} and \texttt{NumPy} \citep{NumPy.software} have been imported using the following Python code:

\begin{lstlisting}[style=Python, caption={Import the \texttt{QMCPy} and \texttt{NumPy} Python packages.}]
import qmcpy as qp 
import numpy as np
\end{lstlisting}

\noindent All code examples and plots are reproducible from a public notebook (\url{https://qmcsoftware.github.io/QMCSoftware/demos/talk_paper_demos/JCAM_Sorokin_2026/JCAM_Sorokin_2026/}). Most of the \texttt{QMCPy} functions presented here are high-level Python interfaces to efficient C implementations in the \texttt{QMCToolsCL} package (\url{https://qmcsoftware.github.io/QMCToolsCL/}). These low-level routines may also be used to build similar interfaces in other programming languages that support C extensions.

\subsection{\texttt{QMCPy} features}

Here we briefly summarize the main \texttt{QMCPy} capabilities relevant to this paper and provide pointers to related software when appropriate.
\begin{description}
    \item[Randomized LD sequences.] \texttt{QMCPy} supports the point sets and randomization routines described below. Subsets of these features are also supported in the comprehensive Java software \texttt{SSJ} \citep{lecuyer.ssj_software} and multi-language \texttt{Magic Point Shop (MPS)} (\url{https://people.cs.kuleuven.be/~dirk.nuyens/qmc-generators/}) \citep{kuo.application_qmc_elliptic_pde}.
    \begin{description}
        \item[Lattice points.] Rank-$1$ lattices may be generated in natural (radical inverse) or linear order, and may be randomized using shifts modulo one. Randomly shifted rank-$1$ lattices are also available in \texttt{MPS} and \texttt{GAIL} (\url{http://gailgithub.github.io/GAIL_Dev/}) (MATLAB's Guaranteed Automatic Integration Library) \citep{GAIL.software,hickernell2018monte}.
        \item[Digital nets.] Digital nets, including higher-order digital nets, may be generated in either natural (digit-reversal) or Gray code order, and may be randomized with linear matrix scrambling (LMS) \citep{owen.variance_alternative_scrambles_digital_net}, digital shifts, digital permutations, or nested uniform scramblings (NUS) \citep{owen1995randomly}. Early implementations of unrandomized digital sequences, including the Faure, Sobol', and Niederreiter--Xing constructions, can be found in \citet{fox1986algorithm}, \citet{bratley1992implementation}, \citet{bratley2003implementing}, and \citet{pirsic2002software}. Considerations for implementing scrambles were discussed by \citet{hong2003algorithm}. Support for combining LMS with digital shifts is also provided in MATLAB, \texttt{MPS}, and both the \texttt{PyTorch} \citep{PyTorch.software} and \texttt{SciPy} \citep{SciPy.software} Python packages.
        \item[Halton points.] As with digital nets, Halton point sets may be randomized with LMS, digital shifts, NUS, and/or permutation scrambles. The implementation of Halton point sets and randomizations have been treated by \citet{owen_halton,wang2000randomized}. The \texttt{QRNG} (\url{https://cran.r-project.org/web/packages/qrng/qrng.pdf}) (Quasi-Random Number Generators) R package \citep{qrng.software} implements generalized Halton point sets \citep{faure2009generalized} which use optimized digital permutation scrambles; these are also supported in \texttt{QMCPy}.
    \end{description}
    The C++ software \texttt{LatNet Builder} \citep{LatNetBuilder.software} as well as the Python software \texttt{QMC4PDE} (\url{https://people.cs.kuleuven.be/~dirk.nuyens/qmc4pde/}) \citep{kuo.application_qmc_elliptic_pde} both provide search routines for finding good lattice generating vectors and digital net generating matrices. \texttt{QMCPy} integrates with the new \textbf{\texttt{LDData}} repository (\url{https://github.com/QMCSoftware/LDData}) which contains a variety of these pregenerated vectors and matrices in standardized formats. \texttt{LDData} additionally includes popular choices from the websites of Frances Kuo on lattices (\url{https://web.maths.unsw.edu.au/~fkuo/lattice/index.html}) \citep{cools2006constructing,nuyens2006fast} and Sobol' points (\url{https://web.maths.unsw.edu.au/~fkuo/sobol/index.html}) (a special case of digital nets) \citep{joe2003remark,joe2008constructing}.
    \item[Variable transformations.] Variable transformations define the distribution of the underlying stochastic variables and automatically rewrite user-defined functions into QMC-compatible forms. \texttt{QMCPy} provides a collection of such transformations, largely wrapping distributions from \texttt{SciPy} \citep{SciPy.software}, so that users can specify integrals and expectations beyond the unit cube without manual change-of-variable calculations.  
    \item[Error estimators for (Q)MC and multilevel (Q)MC.] \texttt{QMCPy} provides several adaptive error estimation algorithms which automatically select the number of points required for a (Q)MC approximation to be within user-specified error tolerances.
    \begin{description}
        \item[Monte Carlo with independent points.] Confidence intervals are derived from either a central limit theorem (CLT) heuristic or a guaranteed version of CLT for functions with bounded kurtosis \citep{hickernell.MC_guaranteed_CI}.
        \item[QMC with multiple randomizations.] Student's-$t$ confidence intervals are formed from independent mean estimates from independent randomizations of an LD point set, following \citet{lecuyer.RQMC_CLT_bootstrap_comparison} and \citet[Chapter 17]{owen.mc_book_practical}. 
        \item[QMC via decay tracking (single LD sequence).]
        Guaranteed error bounds are available for cones of functions whose Fourier or Walsh coefficients decay at predictably rates \citep{hickernell.adaptive_dn_cubature,adaptive_qmc,cubqmclattice,ding2018adaptive}. 
        \item[QMC via Bayesian cubature (single LD sequence).] Posterior credible intervals on the integral of a Gaussian process may be computed efficiently by exploiting fast kernel computations \citep{rathinavel.bayesian_QMC_lattice,rathinavel.bayesian_QMC_sobol,rathinavel.bayesian_QMC_thesis}, as we will describe in this article.
        \item[Multilevel Monte Carlo.] Standard multilevel Monte Carlo (MLMC) \citep{giles.MLMC_path_simulation} and continuation MLMC algorithms \citep{collier2015continuation,robbe2016dimension,robbe2019recycling} are implemented for independent Monte Carlo sampling.  
        \item[Multilevel QMC with multiple randomizations.] Standard multilevel QMC (MLQMC) \citep{giles.mlqmc_path_simulation} and continuation MLQMC algorithms \citep{robbe.multi_index_qmc} are provided, using multiple randomized QMC sequences at each level. 
    \end{description}
    QMC error estimation is treated more broadly in \citet{owen.error_QMC_review}, while \citet{clancy2014cost} and \citet{adaptive_qmc} detail additional considerations for adaptive QMC algorithms. The single-level algorithms were also implemented in \texttt{GAIL}; \texttt{QMCPy} provides Python counterparts and extensions.
    \item[Fast kernel computations.] The second major focus of this paper is fast kernel-based computation using structured LD point sets. Matching a rank-$1$ lattice in natural order to a shift-invariant (SI) RKHS kernel yields a Gram matrix diagonalizable by the FFT and its inverse (IFFT). Similarly, matching a base-$2$ digital net in natural order to a digitally-shift-invariant (DSI) kernel yields a Gram matrix diagonalizable by the fast Walsh--Hadamard transform (FWHT) \citep{fino.fwht}. The  currently supported kernels and fast transforms are described below. Versions of these methods compatible with \texttt{PyTorch} \citep{PyTorch.software} are also maintained in order to enable GPU acceleration. 
    \begin{description}
        \item[SI and DSI kernels.] \texttt{QMCPy} provides SI kernels and DSI kernels, including higher-order smoothness variants. It also provides interfaces to commonly used kernels such as squared exponential and Mat\'ern kernels, which are also available in popular packages such as \texttt{GPyTorch} \citep{gardner.gpytorch_GPU_conjugate_gradient} and \texttt{scikit-learn} \citep{scikit-learn}. SI kernels of arbitrary integer smoothness based on Bernoulli polynomials are well known \citep{kaarnioja.kernel_interpolants_lattice_rkhs,kaarnioja.kernel_interpolants_lattice_rkhs_serendipitous,cools2021fast,cools2020lattice,sloan2001tractability,kuo2004lattice}. DSI kernels of order-$1$ smoothness were derived by \citet{dick.multivariate_integraion_sobolev_spaces_digital_nets}. \citet{baldeaux.polylat_efficient_comp_worse_case_error_cbc} derived expressions for general smoothness parameter $\alpha \geq 2$ and provided a $\calO(\alpha \#x)$ algorithm to compute them, where $\# x$ is the number of nonzero digits in the base-$b$ expansion of $x$. Building on this work, we derive an alternative computable form for the $\alpha=4$ DSI kernel (\Cref{thm:explicit_DSI_low_order_forms}), which involves a rapidly decaying infinite series and can therefore be evaluated efficiently in practice.
        \item[Fast transforms.] \texttt{QMCPy} provides interfaces to the FFT in bit-reversed order (FFTBR), its inverse IFFTBR, and the FWHT algorithms, each with $\calO(n \log n)$ complexity in the number of points $n$. FFTBR and IFFTBR are complexity-preserving wrappers around \texttt{SciPy}'s \citep{SciPy.software} FFT and IFFT functions where the additional bit-reversal step allows for greater compatibility with lattices in natural order. The FWHT is implemented directly and is significantly faster than the implementation in \texttt{SymPy} \citep{10.7717/peerj-cs.103} as of the time of writing; see \Cref{fig:timing}.
    \end{description}
\end{description}

\subsection{Contributions and scope}

This article focuses on equipping \texttt{QMCPy} with expanded support for randomized low-discrepancy (LD) sequences and new tools for fast kernel methods, while providing a unified mathematical and implementation framework. Compared to previous work on \texttt{QMCPy}, the main contributions are to provide:
\begin{description}
    \item[A broadened scope of randomized LD sequences.]
    We extend \texttt{QMCPy} to support a wider class of LD point sets and randomizations, including:
    \begin{itemize}
        \item higher-order digital nets constructed via digital interlacing,
        \item higher-order scrambling of digital nets using either linear matrix scrambling (LMS) or nested uniform scrambling (NUS), and
        \item Halton point sets with LMS- and NUS-based randomizations as well as digital permutations.
    \end{itemize}
    These features are integrated into a consistent Python interface, with careful attention to extensible sequence generation, Gray code ordering, and efficient digital operations.
    \item[Integration with the \texttt{LDData} repository.] \texttt{QMCPy} now interfaces with the \texttt{LDData} repository, which provides collections of rank-$1$ lattice generating vectors and digital net generating matrices in standardized formats. This integration enables users to access high-quality constructions (including higher-order variants) from the literature with minimal effort, and to combine them directly with \texttt{QMCPy}'s randomization and QMC routines.
    \item[SI and DSI kernels of varying smoothness.] 
    \texttt{QMCPy} implements families of shift-invariant (SI) and digitally shift-invariant (DSI) kernels of varying smoothness for use with rank-$1$ lattices and digital nets, respectively, enabling fast kernel-based computations in RKHSs. In particular, we build on existing Walsh-based constructions to provide an alternative computable form of a higher-order $\alpha=4$ DSI kernel (\Cref{thm:explicit_DSI_low_order_forms}) originally studied in \citep{baldeaux.HO_nets_RKHS}; the full derivation and its relation to \citet{baldeaux.polylat_efficient_comp_worse_case_error_cbc} are given in \Cref{sec:DSI_dnets}.
    \item [Fast transforms and efficient Gram matrix eigenvalue updates.]
    By pairing lattices with SI kernels and digital nets with DSI kernels, \texttt{QMCPy} exploits circulant and recursive symmetric block Toeplitz (RSBT) structure in the resulting Gram matrices. We provide an implementation of the fast Walsh--Hadamard transform (FWHT) as well as wrappers of \texttt{SciPy}'s \citep{SciPy.software} FFT and IFFT algorithms called FFTBR and IFFTBR which additionally perform bit-reversals for compatibility with lattices in natural order. Moreover, we develop practical algorithms to update eigenvalues and transformed vectors as the number of points is doubled. These tools reduce the cost of matrix-vector products and linear solves from $\calO(n^2)-\calO(n^3)$ to $\calO(n \log n)$ in $n$, while preserving extensibility of the underlying point sets.
    \item [Numerical validation and practical guidance.]
    We present numerical experiments that (1) compare the cost of generating various randomized LD sequences and applying fast transforms, and (2) study the root-mean-squared error (RMSE) convergence of randomized QMC estimators for a range of test integrands, including classical Genz functions and problems from uncertainty quantification. These results illustrate the strong empirical performance of higher-order digital nets with LMS and digital shifts, the benefits of pairing LD point sets with structured kernels, and the practical trade-offs among different randomizations. Together, they provide guidance on choosing point sets, randomizations, and kernels in applied settings.
\end{description}

Overall, this work positions \texttt{QMCPy} as a practical, reproducible platform that unifies randomized LD sequence generation, adaptive QMC methods, and fast kernel-based computations within a single, well-documented Python library.

\subsection{Outline} 

The remainder of this article is organized as follows. \Cref{sec:notation} gives common notations. \Cref{sec:motivating_problems} describes two motivating problems using low discrepancy sequences: QMC methods (\Cref{sec:QMC}) and fast kernel computations (\Cref{sec:kernel_computations}). \Cref{sec:randomized_LD_seqs} details randomized rank-$1$ lattices (\Cref{sec:lattices}), digital sequences (\Cref{sec:dnets}), and Halton sequences (\Cref{sec:Halton}), then provides a summary of convergence rates along with practical guidance (\Cref{sec:summary_convergence_practical_guidance}). \Cref{sec:fast_transforms} describes pairing SI kernels to rank-$1$ lattices (\Cref{sec:SI_lattices}), DSI kernels to digital nets (\Cref{sec:DSI_dnets}), and the resulting fast kernel computations (\Cref{sec:efficient_updates}). \Cref{sec:numerical_experiments} provides numerical experiments showcasing the speed, accuracy, and versatility of \texttt{QMCPy}. \Cref{sec:conclusion} gives a brief conclusion. 

\section{Notation} \label{sec:notation}

Bold symbols will denote vectors which are assumed to be column vectors, e.g., $\bx \in [0,1]^d$. Collections of $n \in \bbN_0 := \{0,1,2,\dots\}$ objects each in the set $\bbB$ will often be denoted by a vector, e.g., $\bomega \in \bbB^n$ is the vector collecting $n$ real numbers $\omega_i \in \bbB$. The set of all $n \times n'$ matrices with elements in $\bbB$ will be denoted by $\bbB^{n \times n'}$, and we will write matrices as capital letters in serif font, e.g., $\mK \in \bbR^{n \times n}$ is an $n \times n$ matrix. The complex conjugate will be denoted by a bar notation and is understood to act elementwise, e.g., $\mV \in \bbC^{n \times n}$ has complex conjugate $\overline{\mV}$. Lower case letters in serif font denote digits in a base-$b$ expansion, e.g., $i = \sum_{t=0}^{m-1} \mi_t b^t$ for $0 \leq i < b^m$. This may be combined with bold notation when denoting the vector of base-$b$ digits, e.g., $\bmi = (\mi_0,\mi_1,\dots,\mi_{m-1})^\intercal$. Modulo is always taken component-wise, e.g., rank-$1$ lattices will use the notation $\bx \bmod 1 = (x_1 \bmod 1,\dots,x_d \bmod 1)^\intercal$, and digital nets will take all matrix operations to be carried out modulo $b$. Permutations are often denoted by vectors, e.g., $\pi: \{0,1,2\} \to \{0,1,2\}$ with $\pi(0) = 2$, $\pi(1) = 0$, and $\pi(2)=1$ may be denoted by $\pi = (2,0,1)$. 

\section{Motivating problems} \label{sec:motivating_problems} 

\subsection{QMC methods for approximating integrals and expectations} \label{sec:QMC}

Monte Carlo (MC) and quasi-Monte Carlo (QMC) methods approximate a high-dimensional integral over the unit cube by the sample average of function evaluations at certain sampling locations:
\begin{equation}
    \mu := \int_{[0,1]^d} f(\bx) \D \bx \approx \frac{1}{n} \sum_{i=0}^{n-1} f(\bx_i) =: \hmu.
    \label{eq:mc_approx}
\end{equation}
Here $f: [0,1]^d \to \bbR$ is a given integrand and $\{\bx_i\}_{i=0}^{n-1} \subset [0,1]^d$ is a point set. The integral on the left-hand side may be viewed as taking the expectation $\bbE[f(\bX)]$ where $\bX$ is a standard uniform $\bX \sim \calU[0,1]^d$. MC methods choose the sampling locations to be independent and identically distributed (IID) $d$-dimensional standard uniforms $\bx_0,\dots,\bx_{n-1} \simiid \calU[0,1]^d$. IID MC methods for \eqref{eq:mc_approx} have a root-mean-squared-error (RMSE) of $\calO(n^{-1/2})$.

QMC methods \citep{niederreiter.qmc_book,dick.digital_nets_sequences_book,kroese.handbook_mc_methods,dick2022lattice,lemieux2009monte,sloan1994lattice,dick.high_dim_integration_qmc_way} replace IID point sets with LD point sets which more evenly cover the unit cube $[0,1]^d$. This greater uniformity leads to faster convergence of QMC methods compared to IID Monte Carlo methods for regular enough functions. Specifically, for integrands with bounded variation, plugging LD point sets into \eqref{eq:mc_approx} yields a worst-case error rate of $\calO(n^{-1+\delta})$ with $\delta>0$ arbitrarily small compared to the $\calO(n^{-1/2})$ worst-case error rate of IID Monte Carlo methods. 

Randomized quasi-Monte Carlo (RQMC) uses randomized LD point sets to give improved convergence rates and enable practical error estimation. Specifically, if we again assume the integrand has bounded variation, then RQMC methods using digital nets with NUS achieve a RMSE of $\calO(n^{-3/2+\delta})$. Moreover, let $\{\bx_{1i}\}_{i=0}^{n-1}, \dots, \{\bx_{Ri}\}_{i=0}^{n-1}$ denote $R$ IID randomizations of an LD point set where typically $R$ is small, e.g., $R=15$. Then, following  \citet{lecuyer.RQMC_CLT_bootstrap_comparison} and \citet[Chapter 17]{owen.mc_book_practical}, the RQMC estimate
\begin{equation}
    \hmu = \frac{1}{R} \sum_{r=1}^R \hmu_r \qquad \text{where} \quad \hmu_r = \frac{1}{n} \sum_{i=0}^{n-1} f(\bx_{ri})
    \label{eq:RQMC}
\end{equation}
gives a practical $100(1-\tau)\%$ confidence interval $\hmu \pm t_{R-1,\tau/2} \hsigma/\sqrt{R}$ for $\mu$ where $\hsigma^2 = (R-1)^{-1} \sum_{r=1}^R (\hmu_r - \hmu)^2$ and $t_{R-1,\tau/2}$ is the $\tau/2$ quantile of a Student's-$t$ distribution with $R-1$ degrees of freedom. \Cref{code:Genz_ex_1} will show how to code up such an estimate in \texttt{QMCPy}. For rank-$1$ lattices, randomization is typically done using random shifts modulo $1$. For digital nets, randomization is typically done using NUS or the cheaper combination of LMS with digital shifts and/or digit permutations \citep{MATOUSEK1998527,owen.variance_alternative_scrambles_digital_net,owen_halton,owen.gain_coefficients_scrambled_halton}.

Higher-order LD point sets were designed to yield faster convergence for integrands with additional smoothness. For integrands with square integrable mixed partial derivatives up to order $\alpha>1$, plugging higher-order digital nets into $\hmu$ in \eqref{eq:mc_approx} yields a worst-case error rate of $\calO(n^{-\alpha+\delta})$ \citep{dick.walsh_spaces_HO_nets,dick.qmc_HO_convergence_MCQMC2008,dick.decay_walsh_coefficients_smooth_functions}. RQMC using higher-order digital nets with higher-order NUS has been shown to achieve an RMSE of order $\calO(n^{-\alpha-1/2+\delta})$ \citep{dick.higher_order_scrambled_digital_nets}. Lattice points automatically achieve \(\calO(n^{-\alpha+\delta})\) convergence for periodic functions in dominating smoothness Sobolev spaces (Korobov spaces) \citep{dick2022lattice,hickernell1998lattice}. Detailed convergence rates of worst-case errors and RMSEs, both theoretical and empirical, are compared across randomized LD sequences in \Cref{sec:summary_convergence_practical_guidance} along with practical guidance for choosing which randomized LD sequence to use. 

\Cref{fig:point sets} plots some popular LD point sets including rank-$1$ lattices, base-$2$ digital nets (including higher-order versions), and Halton point sets. The plotted lattice is randomly shifted. The plotted digital nets and Halton point sets are randomized with linear matrix scramblings (LMS), digital shifts (DS), digital permutations (DP), and/or nested uniform scramblings (NUS). Digital interlacing of order $\alpha$ is used to construct higher-order randomized digital nets in base-$2$ (DN${}_\alpha$). 

\begin{figure}[ht]
    \centering
    \includegraphics[width=1\textwidth]{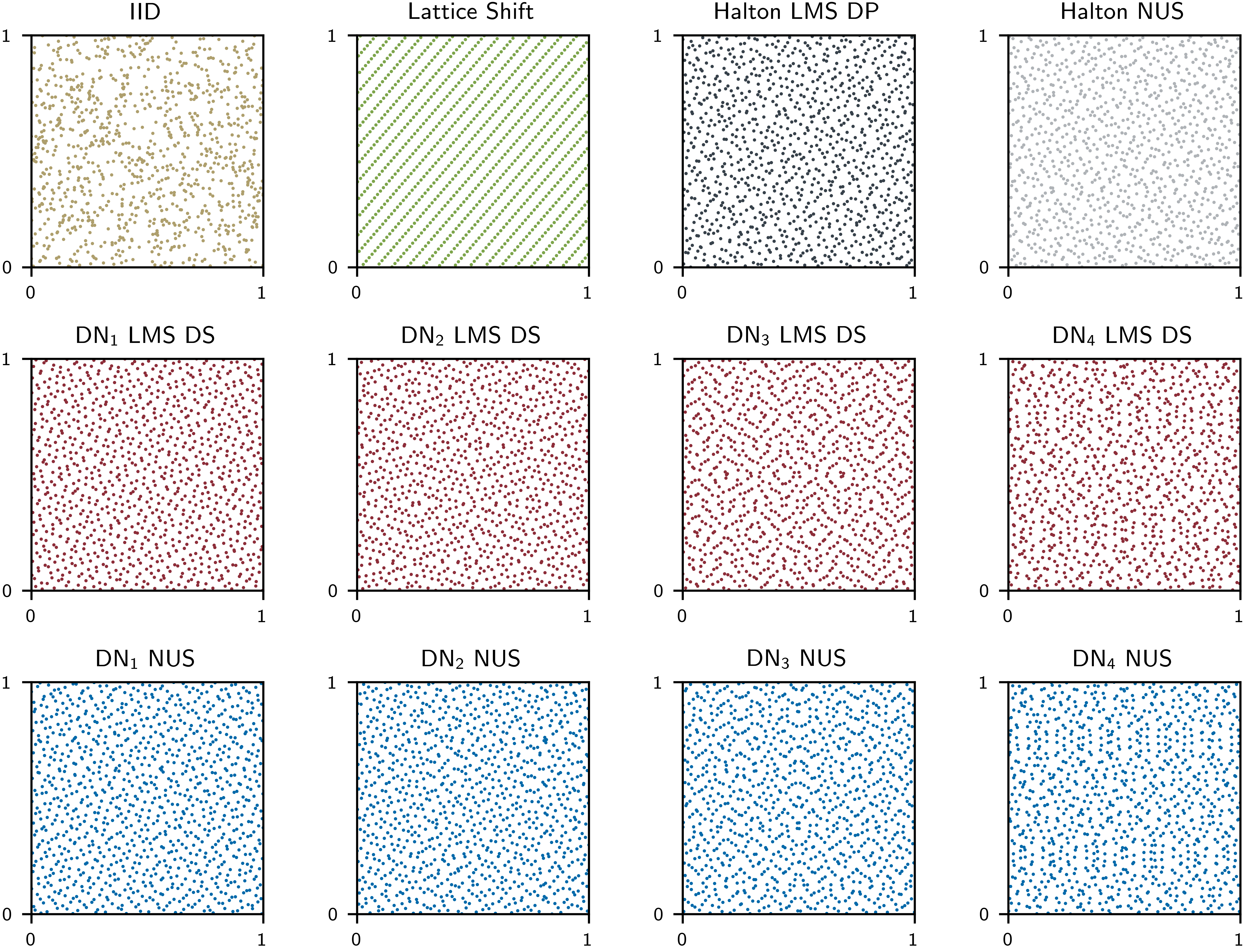}
    \caption{An independent identically distributed (IID) point set and low-discrepancy (LD) point sets of size $n=2^{10}=1024$.}
    \label{fig:point sets}
\end{figure}

\subsection{Kernelized worst-case error and kernel interpolation computations} \label{sec:kernel_computations}

Certain randomized LD point sets may be used to accelerate kernel-based computations. Let $\{\bx_i\}_{i=0}^{n-1}$ be a point set, $K: [0,1]^d \times [0,1]^d \to \bbR$ be a symmetric positive definite (SPD) kernel, and $\mK = \left(K(\bx_i,\bx_{i'})\right)_{i,i'=0}^{n-1}$ be the $n \times n$ SPD Gram matrix of kernel evaluations at all pairs of points. Two motivating kernel computations are described below:

\begin{description}
    \item[Kernelized worst-case error computations.] The error of the approximation \eqref{eq:mc_approx} can be bounded by the product of two terms: a measure of variation of the function $f$ and a measure of the discrepancy of the point set $\{\bx_i\}_{i=0}^{n-1}$. A classical example of this type of bound is the Koksma--Hlawka inequality \citep{hickernell.generalized_discrepancy_quadrature_error_bound,dick.high_dim_integration_qmc_way,hickernell1999goodness,niederreiter.qmc_book}. The discrepancy can be used to design low-discrepancy point sets. \citet{rusch2024message} present a newer idea in this area where good point sets are generated using neural networks trained with a discrepancy-based loss function.  
    
    Let us consider the more general setting of approximating the mean $\mu$ in \eqref{eq:mc_approx} by a weighted cubature scheme $\sum_{i=0}^{n-1} \omega_i f(\bx_i)$. If $f$ is assumed to lie in the RKHS $H$ with kernel $K$, then \citet{hickernell.generalized_discrepancy_quadrature_error_bound} showed that the squared worst-case error equals the generalized kernel discrepancy which takes the form
    $$\int \int K(\bu,\bv) \D \bu \D \bv - 2 \sum_{i=0}^{n-1} \omega_i \int K(\bu,\bx_i) \D \bu + \sum_{i,i'=0}^{n-1} \omega_i \omega_{i'} K(\bx_i,\bx_{i'})$$
    where integrals are understood to be over the unit cube domain $[0,1]^d$. 
    This worst-case error formula serves as the foundation for modern quasi-Monte Carlo (QMC) point construction, including fast component-by-component (CBC) construction algorithms in weighted and higher-order spaces \citep{kuo.application_qmc_elliptic_pde,nuyens2006fast,baldeaux.polylat_efficient_comp_worse_case_error_cbc,cools2021fast,kuo2003component,nuyens2006fast_nonprime,nuyens2004fast,kuo2025constructing,nuyens2014construction}. Evaluating the squared worst-case error formula above requires computing $\mK \bomega$ where $\bomega = \{\omega_i\}_{i=0}^{n-1}$. Moreover, this squared worst-case error is minimized for a given point set by setting the weights to $\bomega^\star = \mK^{-1} \bkappa$ where $\bkappa = \{\kappa_i\}_{i=0}^{n-1}$ collects $\kappa_i = \int_{[0,1]^d} K(\bu,\bx_i) \D \bu$ for $0 \leq i < n$.

    \item[Kernel interpolation computations.] Suppose we would like to approximate $f: [0,1)^d \to \bbR$ given observations $\by = \{y_i\}_{i=0}^{n-1}$ of $f$ at $\{\bx_i\}_{i=0}^{n-1}$ satisfying $y_i = f(x_i)$ for $0 \leq i < n$. Then a kernel interpolant approximates $f$ by $\hf \in H$ where 
    $$\hf(\bx) = \sum_{i=0}^{n-1} \omega_i K(\bx,\bx_i)$$
    and $\bomega = \mK^{-1} \by$. The above kernel interpolant may be reinterpreted as the posterior mean of a Gaussian process regression model. Fitting such a Gaussian process often includes optimizing a kernel's hyperparameters, which may also be done by computing  $\mK \bomega$ and $\mK^{-1} \by$ \citep{rasmussen.gp4ml}. 
\end{description}

Underlying these problems, and many others, is the requirement to compute the matrix-vector product $\mK \by$ or solve the linear system $\mK^{-1}\by$. The standard costs of these computations are $\calO(n^2)$ and $\calO(n^3)$ respectively. One method to reduce these high costs is to induce structure into $\mK$. \citet{zeng.spline_lattice_digital_net,zeng.spline_lattice_error_analysis} proposed structure-inducing methods which use certain kernels $K$ and points $\{\bx_i\}_{i=0}^{n-1}$ which enable both computations to be done at only $\calO(n \log n)$ cost. Two such pairings exist:
\begin{enumerate}
    \item When a rank-$1$ lattice point set $\{\bx_i\}_{i=0}^{n-1}$ is paired with a SI kernel, $\mK$ has a circulant structure and thus matrix multiplication and linear system solutions in $\mK$ may be computed by performing FFTs and IFFTs. For rank-$1$ lattices in natural order, the FFT and IFFT with bit-reversals, called FFTBR and IFFTBR respectively, are used. 
    \item When a base-$2$ digital net $\{\bx_i\}_{i=0}^{n-1}$ is paired with a DSI kernel, $\mK$ has a recursive symmetric block Toeplitz (RSBT) structure and thus matrix multiplication and linear system solutions in $\mK$ may be computed by only performing FWHTs.
\end{enumerate}


\section{Randomized low-discrepancy sequences} \label{sec:randomized_LD_seqs}

\subsection{Common definitions}

For a fixed base $b \in \bbN$, write $i \in \bbN_0$ as $i = \sum_{t=0}^\infty \mi_t b^t$ so $\mi_t$ is the $t^\mathrm{th}$ digit in the base-$b$ expansion of $i$. We denote the vector of the first $m$ base-$b$ digits of $i$ by
$$\bD_m(i) = (\mi_0,\mi_1,\dots,\mi_{m-1})^\intercal.$$
For $\bmi = \bD_m(i)$, we denote the digit-reversal of $i$ by
$$F_m(\bmi) = \sum_{t=1}^m \mi_{t-1} b^{-t} \in [0,1)$$
Finally, let
\begin{equation}
    v(i) = \sum_{t=1}^\infty \mi_{t-1} b^{-t}
    \label{eq:v}
\end{equation}
so that $v(i) = F_m(\bD_m(i))$ when $0 \leq i < b^m$, which is always the case in this article. The van der Corput sequence in base $b$ is $\{v(i)\}_{i \geq 0}$. 

\subsection{Rank-$1$ lattices} \label{sec:lattices}

Consider a fixed \emph{generating vector} $\bg \in \bbN^d$ and fixed prime base $b$. Then we define the \emph{lattice sequence} 
\begin{equation}
    \bz_i = v(i) \bg \bmod 1, \qquad i \geq 0.
    \label{eq:unrandomized_lattice}
\end{equation}
If $n=b^m$ for some $m \in \bbN_0$, then the lattice point set $\{\bz_i\}_{i=0}^{n-1} \subset [0,1)^d$ is equivalent to $\{i \bg/n \bmod 1\}_{i=0}^{n-1}$ where we say the former is in \emph{natural (radical inverse) order} while the latter is in \emph{linear order}. When $n$ is not of the form $b^m$, then linear and natural order will not generate the same lattice point set. Linear order assumes the number of points $n$ is fixed, while natural order is extensible in $n$, i.e., increasing $n$ through powers of $b$ simply extends the sequence without discarding previous points. On the other hand, linear order assumes $n$ is fixed, so increasing $n$ through powers of $b$ will not share points, e.g., in linear order with $b=2$ the first $4$ points when $n=4$ are different from the first $4$ points when $n=8$. Our implementation supports linear order for any base $b$, and natural order only for base $b=2$. 

\begin{description}
    \item[Shifted lattice] For a shift $\bDelta \in [0,1)^d$, we define the shifted point set 
    $$\bx_i= (\bz_i + \bDelta) \bmod 1$$ 
\end{description}
where $\bz_i$ denotes the unrandomized lattice point in \eqref{eq:unrandomized_lattice}. Randomized lattices use $\bDelta \sim \calU[0,1]^d$. In the following code snippet we generate $R$ independently shifted lattices with shifts $\bDelta_1,\dots,\bDelta_R \simiid \calU[0,1]^d$. 

\begin{lstlisting}[style=Python, caption={Generate randomized lattices.}, label={code:lattice}]
lattice = qp.Lattice(
    dimension = 52,
    randomize = "shift", # or None for unrandomized
    replications = 16, # R
    order = "natural", # or "linear"
    seed = None, # pass integer seed for reproducibility
    generating_vector = "mps.exod2_base2_m20_CKN.txt")
x = lattice(2**10) # (16, 1024, 52) shaped numpy.ndarray
\end{lstlisting}

\noindent Here we have used a generating vector from \citet{cools2006constructing} which is stored in a standardized format in the \texttt{LDData} repository. Other generating vectors from \texttt{LDData} may be used by passing in a file name from \url{https://github.com/QMCSoftware/LDData/tree/main/lattice/} or by passing in an explicit array. One may generate a subset of the sequence by passing in \texttt{n\_min} (inclusive) and \texttt{n\_max} (exclusive) parameters, e.g., \texttt{x = lattice(n\_min=8, n\_max=16)} will generate an array of shape \texttt{(16, 8, 52)} of points at indices $i=8,9,\dots,15$ in natural order. This method of generating subsequences will also apply for the digital nets example in \Cref{code:dnb2} and the Halton example in \Cref{code:Halton}.  

One often transforms lattices using the baker transform 
\begin{equation} \label{eq:baker}
    \bx \mapsto 1-2 \lvert \bx - 1/2 \rvert
\end{equation}
where the absolute value is applied elementwise. This periodizes the function and can improve convergence rates for non-periodic integrands; for functions with dominating mixed smoothness two, tent-transformed lattices achieve second-order convergence, See \Cref{sec:summary_convergence_practical_guidance} and the references therein for details and \Cref{sec:numerical_experiments} for numerical experiments in which the baker transform is used to achieve improved convergence. 

\subsection{Digital nets} \label{sec:dnets}

Consider a fixed prime base $b$ and \emph{generating matrices} $\mC_1,\dots,\mC_d \in \{0,...,b-1\}^{t_\mathrm{max} \times m}$ where $t_\mathrm{max} \in \bbN$ and $m \in \bbN$ are fixed and $b$ is a given prime base. Our implementation supports any prime base $b$, with extensions to theoretically-supported non-prime bases left as future work. The first $n=b^m$ points of a digital sequence form a \emph{digital net} $\{\bz_i\}_{i=0}^{b^m-1} \subset [0,1)^d$ where, in \emph{natural (digit-reversal) order}, for $1 \leq j \leq d$ and $0 \leq i < b^m$ with base-$b$ digit vector $\bmi = \bD_m(i)$, 
$$z_{ij} = F_{t_\mathrm{max}}(\bmz_{ij}) \qquad\text{with}\qquad \bmz_{ij} = \mC_j \bmi \bmod b \in \{0,\dots,b-1\}^{t_\mathrm{max}}.$$

\begin{description}
    \item[Digital shifts (DSs).] Similar to lattices, one may apply a shift $\mDelta \in [0,1)^{t_\mathrm{max} \times d}$ to the digital net to get a digitally-shifted digital net $\{\bx_i\}_{i=0}^{b^m-1}$ where 
    $$x_{ij} = F_{t_\mathrm{max}}((\bmz_{ij} + \bDelta_j ) \bmod b)$$
    and $\bDelta_j$ is the $j^\mathrm{th}$ column of $\mDelta$. Digital nets with random digital shifts use $\mDelta \simiid \calU\{0,\dots,b-1\}^{t_\mathrm{max} \times d}$ with IID uniform entries from $\{0,\dots,b-1\}$. 
    \item[Digital permutation (DPs).] In what follows we will denote permutations of $\{0,\dots,b-1\}$ by $\pi$. Suppose we are given a set of permutations
    $$\mPi = \{\pi_{j,t}: \quad 1 \leq j \leq d, \quad 0 \leq t < t_\mathrm{max}\}.$$
    Then we may construct the digitally permuted digital net $\{\bx_i\}_{i=0}^{b^m-1}$ where 
    $$x_{ij} = F_{t_\mathrm{max}}((\pi_{j,0}(\mz_{ij0}),\pi_{j,1}(\mz_{ij1}),\dots,\pi_{j,t_\mathrm{max}-1}(\mz_{ij(t_\mathrm{max}-1)}))^\intercal).$$
    Randomly permuted digital nets use independent permutations chosen uniformly over all permutations of $\{0,\dots,b-1\}$. 
    \item[Nested Uniform Scrambling (NUS).] NUS is often called Owen scrambling for its conception by \citet{owen1995randomly}. As before, $\pi$ denotes permutations of $\{0,\dots,b-1\}$. Now suppose we are given a set of permutations 
    $$\calP = \{\pi_{j,\mv_1\cdots\mv_t}: \, 1 \leq j \leq d, \, 0 \leq t < t_\mathrm{max}, \, \mv_k \in \{0,\dots,b-1\}, \, 0 \leq k \leq t \}.$$
    Then a \emph{nested uniform scrambling} of a digital net is $\{\bx_i\}_{i=0}^{b^m-1}$ where
    $$x_{ij} = F_{t_\mathrm{max}}\left(\begin{bmatrix} \pi_{j,}(\mz_{ij0}) \\ \pi_{j,\mz_{ij0}}(\mz_{ij1}) \\ \pi_{j,\mz_{ij0}\mz_{ij1}}(\mz_{ij2}) \\ \vdots \\ \pi_{j,\mz_{ij0}\mz_{ij1}\cdots\mz_{ij(t_\mathrm{max}-2)}}(\mz_{ij(t_\mathrm{max}-1)})\end{bmatrix}\right).$$ 
    As the number of elements in $\calP$ is  
    $$\lvert \calP \rvert = d(1+b+b^2+\dots+b^{t_\mathrm{max}-1}) = d \frac{b^{t_\mathrm{max}}-1}{b-1},$$ 
    NUS implementations cannot afford to generate all permutations a priori. In \texttt{QMCPy}, permutations are generated and stored only as needed. At the time of writing, this is accomplished with an adaptively built search tree where each node either stores the permutation for that digit if it has children or stores permutations for all remaining digits if it does not have children. As with digitally-permuted digital nets, NUS uses independent uniform random permutations.  
    \item[Linear matrix scrambling (LMS).] LMS is a computationally cheaper version of NUS which has empirically shown the same improved convergence rate of NUS for many practical problems. LMS uses scrambling matrices $\mS_1,\dots,\mS_d \in \{0,\dots,b-1\}^{t_\mathrm{max} \times t_\mathrm{max}}$ and sets the LMS generating matrices $\tmC_1,\dots,\tmC_d \in \{0,\dots,b-1\}^{t_\mathrm{max} \times m}$ so that 
    $$\tmC_j = \mS_j \mC_j \bmod b$$
    for $j=1,\dots,d$. Following \citet{owen.variance_alternative_scrambles_digital_net}, let us denote elements in $\{1,\dots,b-1\}$ by $h$ and elements in $\{0,\dots,b-1\}$  by $g$. Then common structures for $\mS_j$ include lower-triangular digit matrices with positive diagonals
    \begin{align*}
        &\begin{bmatrix} h_{11} & 0 & 0 & 0 & \dots \\ g_{21} & h_{22} & 0 & 0 & \dots \\ g_{31} & g_{32} & h_{33} & 0 & \dots \\ g_{41} & g_{42} & g_{43} & h_{44} & \dots \\ \vdots & \vdots & \vdots & \vdots & \ddots \end{bmatrix}, \\
        &\begin{bmatrix} h_1 & 0 & 0 & 0 & \dots \\ g_2 & h_1 & 0 & 0 & \dots \\ g_3 & g_2 & h_1 & 0 & \dots \\ g_4 & g_3 & g_2 & h_1 & \dots \\ \vdots & \vdots & \vdots & \vdots & \ddots \end{bmatrix}, \quad\text{and}\quad \\
        &\begin{bmatrix} h_1 & 0 & 0 & 0 & \dots \\ h_1 & h_2 & 0 & 0 & \dots \\ h_1 & h_2 & h_3 & 0 & \dots \\ h_1 & h_2 & h_3 & h_4 & \dots \\ \vdots & \vdots & \vdots & \vdots & \ddots \end{bmatrix}
    \end{align*}
    which corresponds respectively to Matou\v{s}ek's linear scrambling \citep{MATOUSEK1998527}, Tezuka's $i$-binomial scrambling \citep{tezuka2002randomization}, and Owen's striped LMS \citep{owen.variance_alternative_scrambles_digital_net} (not to be confused with NUS which is often called Owen scrambling). These three LMS methods choose $g$ and $h$ values all independently and uniformly from $\{1,\dots,b-1\}$ and $\{0,\dots,b-1\}$ respectively. \citet{owen.variance_alternative_scrambles_digital_net} analyzes these scramblings and their connection to NUS.
    \item[Digital interlacing for higher-order nets.] Digital interlacing enables the construction of higher-order digital nets. For integer order $\alpha \geq 1$, interlacing $\mA_1,\mA_2,\dots \in \{0,\dots,b-1\}^{t_\mathrm{max} \times m}$ gives $\widehat{\mA}_1,\widehat{\mA}_2,\dots \in \{0,\dots,b-1\}^{\alpha t_\mathrm{max} \times m}$ satisfying $\widehat{\mA}_{jtk} = A_{\widehat{j},\widehat{t},k}$ where $\widehat{j} = \alpha (j-1) + (t \bmod \alpha) + 1$ and $\widehat{t} = \lfloor t / \alpha \rfloor$ for $j \geq 1$ and $0 \leq t < \alpha t_\mathrm{max}$ and $1 \leq k \leq m$. Here $\widehat{\mA}_{jtk}$ is the element in row $t$ and column $k$ of matrix $\widehat{\mA}_j$ and similarly for $\mA_{\widehat{j},\widehat{t},k}$. For example, with $m=2$, $t_\mathrm{max}=2$, and $\alpha=2$ we have
    $$\mA_1 = \begin{bmatrix} a_{101} & a_{102} \\ a_{111} & a_{112} \end{bmatrix}, \quad \mA_2 = \begin{bmatrix} a_{201} & a_{202} \\ a_{211} & a_{212} \end{bmatrix}, \quad \longrightarrow\quad  \widehat{\mA}_1 = \begin{bmatrix} a_{101} & a_{102} \\ a_{201} & a_{202} \\ a_{111} & a_{112} \\ a_{211} & a_{212} \end{bmatrix}$$
    $$\mA_3 = \begin{bmatrix} a_{301} & a_{302} \\ a_{311} & a_{312} \end{bmatrix}, \quad \mA_4 = \begin{bmatrix} a_{401} & a_{402} \\ a_{411} & a_{412} \end{bmatrix} \quad\longrightarrow\quad  \widehat{\mA}_2 = \begin{bmatrix} a_{301} & a_{302} \\ a_{401} & a_{402} \\ a_{311} & a_{312} \\ a_{411} & a_{412}\end{bmatrix}$$
    where $a_{jtk} \in \{0,\dots,b-1\}$. Without scrambling, higher-order digital nets may be directly generated from $\widehat{\mC}_1,\dots,\widehat{\mC}_d$, the interlaced generating matrices resulting from interlacing $\mC_1,\dots,\mC_{\alpha d}$. Higher-order NUS \citep{dick.higher_order_scrambled_digital_nets} requires generating a digital net from $\mC_1,\dots,\mC_{\alpha d}$, applying NUS to the resulting $\alpha d$-dimensional digital net, and then interlacing the resulting digit matrices $\{\bmz_{i,1}^\intercal\}_{i=0}^{b^m-1},\dots,\{\bmz_{i,\alpha d}^\intercal\}_{i=0}^{b^m-1} \in \{0,\dots,b-1\}^{t_\mathrm{max} \times m}$. For higher-order LMS with interlacing, one applies LMS to the generating matrices $\mC_1,\dots,\mC_{\alpha d}$ to get $\tmC_1,\dots,\tmC_{\alpha d}$, then interlaces $\tmC_1,\dots,\tmC_{\alpha d}$ to get $\widehat{\tmC}_1,\dots,\widehat{\tmC}_d$, then generates the digital net from $\widehat{\tmC}_1,\dots,\widehat{\tmC}_d$. As we show in the numerical experiments in \Cref{sec:numerical_experiments}, LMS is significantly faster than NUS (especially for higher-order nets) while still achieving higher-order rates of RMSE convergence. 
\end{description}

A subtle difference between the above presentation and practical implementation is that $t_\mathrm{max}$ may change with randomization in practice. For example, suppose we are given generating matrices $\mC_1,\dots,\mC_d \in \{0,\dots,b-1\}^{32 \times32}$ but would like the shifted digital net to have $64$ digits of precision. Then we should generate $\bDelta \in \{0,\dots,b-1\}^{t_\mathrm{max} \times d}$ with $t_\mathrm{max}=64$ and treat $\mC_j$ as $t_\mathrm{max} \times t_\mathrm{max}$ matrices with appropriate trailing rows and columns zeroed out. 

Gray code ordering of digital nets enables computing the next point $\bx_{i+1}$ from $\bx_i$ by only adding a single column from each generating matrix. Specifically, the $q^\mathrm{th}$ column of each generating matrix gets digitally added to the previous point where $q-1$ is the index of the only digit to be changed in Gray code ordering, i.e., in Gray code order only one digit is incremented or decremented by $1$ (modulo $b$) when $i$ is incremented by $1$. Gray code orderings for $b=2$ and $b=3$ are shown in \Cref{tab:Gray code}. 

\begin{table}[htbp]
    \centering
    \begin{tabular}{r | r l | r l }
        $i$ & $i_2$ & Gray code  $i_2$ & $i_3$ & Gray code $i_3$ \\
        \hline 
        $0$ & $0000_2$ & $0000_2$ & $00_3$ & $00_3$ \\
        $1$ & $0001_2$ & $0001_2$ & $01_3$ & $01_3$ \\ 
        $2$ & $0010_2$ & $0011_2$ & $02_3$ & $02_3$ \\ 
        $3$ & $0011_2$ & $0010_2$ & $10_3$ & $12_3$ \\ 
        $4$ & $0100_2$ & $0110_2$ & $11_3$ & $11_3$ \\ 
        $5$ & $0101_2$ & $0111_2$ & $12_3$ & $10_3$ \\ 
        $6$ & $0110_2$ & $0101_2$ & $20_3$ & $20_3$ \\ 
        $7$ & $0111_2$ & $0100_2$ & $21_3$ & $21_3$ \\ 
        $8$ & $1000_2$ & $1000_2$ & $22_3$ & $22_3$
    \end{tabular}
    \caption{Gray code order for bases $b=2$ and $b=3$.}
    \label{tab:Gray code} 
\end{table}

The following code generates $\alpha=2$ higher-order digital nets in base-$2$ with $R$ independent LMS and digital shift (DS) combinations. 

\begin{lstlisting}[style=Python, caption={Generate randomized higher-order digital nets in base $b=2$.}, label={code:dnb2}]
dnb2 = qp.DigitalNetB2(
    dimension = 52, 
    randomize = "LMS DS", # Matousek's LMS + DS
    # or ["NUS","DS","LMS",None]
    t = 64, # number of LMS bits, i.e., rows in S_j
    alpha = 2, # interlacing factor for higher order nets
    replications = 16, # R
    order = "natural", # or "Gray"
    seed = None, # integer seed for reproducibility
    generating_matrices = "joe_kuo.6.21201.txt")
x = dnb2(2**10) # (16, 1024, 52) shaped numpy.ndarray
\end{lstlisting}

\noindent Here we have used a set of Sobol' generating matrices from \citet{joe2008constructing}, which are stored in a standardized format in the \texttt{LDData} repository. Specifically, we use the ``new-joe-kuo-6.21201'' direction numbers from \url{https://web.maths.unsw.edu.au/~fkuo/sobol/index.html}. Other generating matrices from \texttt{LDData} may be used by passing in a file name from \url{https://github.com/QMCSoftware/LDData/blob/main/dnet/} or by passing an explicit array.

Our \texttt{QMCPy} implementation also supports digital nets with prime bases other than $2$. For example, the following code generates a Faure sequence with $R$ independent LMS and digital permutation (DP) combinations. Faure sequences use a prime base equal to the smallest prime greater than or equal to the dimension, e.g., the $d=52$ dimensional Faure sequence uses prime base $b=53$ in the following setup. 

\begin{lstlisting}[style=Python, caption={Generate randomized Faure points in base $b=53$.}, label={code:faure}]
faure = qp.Faure(
    dimension = 52, 
    randomize = "LMS DP", # Matousek's LMS + DP
    # or ["LMS DS","LMS","DP","DS","NUS",None]
    replications = 16, # R
    seed = None) # pass integer seed for reproducibility
x = faure(53**2) # (16, 2809, 52) shaped numpy.ndarray
\end{lstlisting}

\subsection{Halton sequences} \label{sec:Halton}

The digital sequences described in the previous section used a fixed prime base $b$. One may allow each dimension $j \in \{1,\dots,d\}$ to have its own prime base $b_j$. The most popular of such constructions is the Halton sequence which sets $b_j$ to the $j^\mathrm{th}$ prime, sets $\mC_j$ to the identity matrix, and sets $t_\mathrm{max} = m$. This enables the simplified construction of Halton points $\{\bx_i\}_{i=0}^{n-1}$ via
$$\bx_i = (v_{b_1}(i),\dots,v_{b_d}(i))^\intercal$$
where we have added a subscript to $v$ in \eqref{eq:v} to denote the base dependence. 

Almost all the methods described for digital sequences are immediately applicable to Halton sequences after accounting for the differing bases across dimensions. However, digital interlacing is not generally applicable when the bases differ. Halton with random starting points has also been explored by \citet{wang2000randomized}, although we do not consider this here.
The following code generates a Halton point set with $R$ independent LMS and digital permutation combinations. Specifically, we generate $R$ independent LMS-Halton point sets and then apply an independent permutation scramble to each.

\begin{lstlisting}[style=Python, caption={Generate randomized Halton points.}, label={code:Halton}]
halton = qp.Halton(
    dimension = 52, 
    randomize = "LMS DP", # Matousek's LMS + DP
    # or ["LMS DS","LMS","DP","DS","NUS","QRNG",None]
    t = 64, # number of LMS digits, i.e.,  rows in S_j
    replications = 16, # R
    seed = None) #  pass integer seed for reproducibility
x = halton(2**10) # (16, 1024, 52) shaped numpy.ndarray
\end{lstlisting}

\noindent The \texttt{"QRNG"} randomization option  follows the \texttt{QRNG} software package \citep{qrng.software} in generating a generalized Halton point set \citep{faure2009generalized} which uses a deterministic set of  permutation scrambles and random digital shifts. 

\subsection{Summary of randomizations, convergence rates, and practical guidance} \label{sec:summary_convergence_practical_guidance}

\begin{figure}[htbp]
    \centering
    \newcommand{\scale}{5}
    \newcommand{\axsep}{1.5}
    \begin{subfigure}[t]{.2\textwidth}
    \begin{tikzpicture}
        \draw (0,\scale+1) node{$x$};
        \draw (\axsep,\scale+1) node{$z$};
        \draw[-,color=lightgray] (0,0) -- (\axsep,0);
        \draw[-,color=lightgray] (0,\scale) -- (\axsep,\scale);
        \draw[-] (0,0) -- (0,\scale);
        \draw[-] (\axsep,0) -- (\axsep,\scale); 
        \draw[-] (0,0) -- (.4,0);
        \draw[-] (0,\scale) -- (.4,\scale);
        \draw[-] (\axsep-.4,0) -- (\axsep,0);
        \draw[-] (\axsep-.4,\scale) -- (\axsep,\scale);
        \draw (\axsep/2,\scale+2) node{lattice + shift};
        \draw (1,\scale-\scale*\thisDelta) node{$1-\Delta$};
        \draw (\axsep-1,\scale*\thisDelta) node{$\Delta$};
        \draw[-] (0,\scale-\scale*\thisDelta) -- (.4,\scale-\scale*\thisDelta);
        \draw[-] (\axsep-.4,\scale*\thisDelta) -- (\axsep,\scale*\thisDelta);
        \draw[->] (0,\scale/2-\scale*\thisDelta/2) -- (\axsep,\scale-\scale/2+\scale*\thisDelta/2);
        \draw[-] (0,\scale-\scale*\thisDelta/2)-- (\axsep/2,\scale);
        \draw[->] (\axsep/2,0)-- (\axsep,\scale*\thisDelta/2);
    \end{tikzpicture}
    \end{subfigure}
    \begin{subfigure}[b]{.25\textwidth}
    \begin{tikzpicture}
        \draw (0,\scale+1) node{$x$};
        \draw (2*\axsep,\scale+1) node{$z$};
        \draw[-] (0,0) -- (0,\scale);
        \draw[-] (\axsep,0) -- (\axsep,\scale); 
        \draw[-] (2*\axsep,0) -- (2*\axsep,\scale); 
        \draw[-,color=lightgray] (0,0) -- (\axsep,0);
        \draw[-,color=lightgray] (0,\scale) -- (\axsep,\scale);
        \foreach \x in {0,...,3}{
            \draw[-,color=lightgray] (\axsep,\scale*\x/3) -- (2*\axsep,\scale*\x/3);
            \draw[-] (0,\scale*\x/3) -- (.4,\scale*\x/3);
            \draw[-] (\axsep-.4,\scale*\x/3) -- (\axsep+.4,\scale*\x/3);
            \draw[-] (2*\axsep-.4,\scale*\x/3) -- (2*\axsep,\scale*\x/3);
        }
        \foreach \x in {0,...,9}{
            \draw[-] (\axsep,\scale*\x/9) -- (\axsep+.2,\scale*\x/9);
            \draw[-] (2*\axsep-.2,\scale*\x/9) -- (2*\axsep,\scale*\x/9);
        }
        \draw (\axsep,\scale+2) node{digital net + DS};
        \draw[->] (0,\scale*1.5/9)-- (\axsep,\scale*7.5/9);
        \draw[-] (0,\scale*4.5/9)-- (\axsep*3/4,\scale);
        \draw[->] (\axsep*3/4,0)-- (\axsep,\scale*1.5/9);
        \draw[-] (0,\scale*7.5/9)-- (\axsep/4,\scale);
        \draw[->] (\axsep/4,0)-- (\axsep,\scale*4.5/9);
        \draw[->] (\axsep,\scale/18) -- (2*\axsep,\scale*3/18);
        \draw[->] (\axsep,\scale*3/18) -- (2*\axsep,\scale*5/18);
        \draw[-] (\axsep,\scale*5/18) -- (\axsep+\axsep/2,\scale/3);
        \draw[->] (\axsep+\axsep/2,0) -- (2*\axsep,\scale/18);
        \draw[->] (\axsep,\scale/3+\scale/18) -- (2*\axsep,\scale/3+\scale*3/18);
        \draw[->] (\axsep,\scale/3+\scale*3/18) -- (2*\axsep,\scale/3+\scale*5/18);
        \draw[-,densely dotted] (\axsep,\scale/3+\scale*5/18) -- (\axsep+\axsep/2,\scale/3+\scale/3);
        \draw[->,densely dotted] (\axsep+\axsep/2,\scale/3) -- (2*\axsep,\scale/3+\scale/18);
        \draw[->] (\axsep,\scale*2/3+\scale/18) -- (2*\axsep,\scale*2/3+\scale*3/18);
        \draw[->] (\axsep,\scale*2/3+\scale*3/18) -- (2*\axsep,\scale*2/3+\scale*5/18);
        \draw[-] (\axsep,\scale*2/3+\scale*5/18) -- (\axsep+\axsep/2,\scale*2/3+\scale/3);
        \draw[->] (\axsep+\axsep/2,\scale*2/3) -- (2*\axsep,\scale*2/3+\scale/18);
    \end{tikzpicture}
    \end{subfigure}
    \begin{subfigure}[b]{.25\textwidth}
    \begin{tikzpicture}
        \draw (0,\scale+1) node{$x$};
        \draw (2*\axsep,\scale+1) node{$z$};
        \draw[-] (0,0) -- (0,\scale);
        \draw[-] (\axsep,0) -- (\axsep,\scale); 
        \draw[-] (2*\axsep,0) -- (2*\axsep,\scale);
        \draw[-,color=lightgray] (0,0) -- (\axsep,0);
        \draw[-,color=lightgray] (0,\scale) -- (\axsep,\scale); 
        \foreach \x in {0,...,3}{
            \draw[-,color=lightgray] (\axsep,\scale*\x/3) -- (2*\axsep,\scale*\x/3);
            \draw[-] (0,\scale*\x/3) -- (.4,\scale*\x/3);
            \draw[-] (\axsep-.4,\scale*\x/3) -- (\axsep+.4,\scale*\x/3);
            \draw[-] (2*\axsep-.4,\scale*\x/3) -- (2*\axsep,\scale*\x/3);
        }
        \foreach \x in {0,...,9}{
            \draw[-] (\axsep,\scale*\x/9) -- (\axsep+.2,\scale*\x/9);
            \draw[-] (2*\axsep-.2,\scale*\x/9) -- (2*\axsep,\scale*\x/9);
        }
        \draw (\axsep,\scale+2) node{digital net + DP};
        \draw[->] (0,\scale*1/6) -- (\axsep,\scale*5/6);
        \draw[->] (0,\scale*3/6) -- (\axsep,\scale*1/6);
        \draw[->] (0,\scale*5/6) -- (\axsep,\scale*3/6);
        \draw[->] (\axsep,\scale*1/18) -- (2*\axsep,\scale*5/18);
        \draw[->] (\axsep,\scale*3/18) -- (2*\axsep,\scale*3/18);
        \draw[->] (\axsep,\scale*5/18) -- (2*\axsep,\scale*1/18);
        \draw[->] (\axsep,\scale*7/18) -- (2*\axsep,\scale*11/18);
        \draw[->] (\axsep,\scale*9/18) -- (2*\axsep,\scale*9/18);
        \draw[->] (\axsep,\scale*11/18) -- (2*\axsep,\scale*7/18);
        \draw[->] (\axsep,\scale*13/18) -- (2*\axsep,\scale*17/18);
        \draw[->] (\axsep,\scale*15/18) -- (2*\axsep,\scale*15/18);
        \draw[->] (\axsep,\scale*17/18) -- (2*\axsep,\scale*13/18);
    \end{tikzpicture}
    \end{subfigure}
    \begin{subfigure}[b]{.25\textwidth}
    \begin{tikzpicture}
        \draw (0,\scale+1) node{$x$};
        \draw (2*\axsep,\scale+1) node{$z$};
        \draw[-] (0,0) -- (0,\scale);
        \draw[-] (\axsep,0) -- (\axsep,\scale); 
        \draw[-] (2*\axsep,0) -- (2*\axsep,\scale); 
        \draw[-,color=lightgray] (0,0) -- (\axsep,0);
        \draw[-,color=lightgray] (0,0+\scale) -- (\axsep,0+\scale);
        \foreach \x in {0,...,3}{
            \draw[-,color=lightgray] (\axsep,\scale*\x/3) -- (2*\axsep,\scale*\x/3);
            \draw[-] (0,\scale*\x/3) -- (.4,\scale*\x/3);
            \draw[-] (\axsep-.4,\scale*\x/3) -- (\axsep+.4,\scale*\x/3);
            \draw[-] (2*\axsep-.4,\scale*\x/3) -- (2*\axsep,\scale*\x/3);
        }
        \foreach \x in {0,...,9}{
            \draw[-] (\axsep,\scale*\x/9) -- (\axsep+.2,\scale*\x/9);
            \draw[-] (2*\axsep-.2,\scale*\x/9) -- (2*\axsep,\scale*\x/9);
        }
        \draw (\axsep,\scale+2) node{digital net + NUS};
        \draw[->] (0,\scale*1/6) -- (\axsep,\scale*5/6);
        \draw[->] (0,\scale*3/6) -- (\axsep,\scale*1/6);
        \draw[->] (0,\scale*5/6) -- (\axsep,\scale*3/6);
        \draw[->] (\axsep,\scale*1/18) -- (2*\axsep,\scale*1/18);
        \draw[->] (\axsep,\scale*3/18) -- (2*\axsep,\scale*3/18);
        \draw[->] (\axsep,\scale*5/18) -- (2*\axsep,\scale*5/18);
        \draw[->] (\axsep,\scale*7/18) -- (2*\axsep,\scale*9/18);
        \draw[->] (\axsep,\scale*9/18) -- (2*\axsep,\scale*11/18);
        \draw[->] (\axsep,\scale*11/18) -- (2*\axsep,\scale*7/18);
        \draw[->] (\axsep,\scale*13/18) -- (2*\axsep,\scale*17/18);
        \draw[->] (\axsep,\scale*15/18) -- (2*\axsep,\scale*15/18);
        \draw[->] (\axsep,\scale*17/18) -- (2*\axsep,\scale*13/18);
    \end{tikzpicture}
    \end{subfigure}
    \caption{Low-discrepancy randomization routines in dimension $d=1$.} 
    \label{fig:ld_randomizations}
\end{figure}

\begin{table}[htbp]
    \centering
    \resizebox{\textwidth}{!}{%
    \begin{tabular}{lllll}
        sequence / randomization & WCE & RMSE (theory) & RMSE (empirical) & references \\
        \hline 
        rank-$1$ lattice (unrandomized) & $\calO(n^{-\alpha+\delta})$ & -- & -- & \citep{dick.higher_order_scrambled_digital_nets,sloan2001tractability,kuo2004lattice,dick2022lattice,sloan1994lattice} \\
        rank-$1$ lattice + random shift  & $\calO(n^{-\alpha+\delta})$ & $\calO(n^{-\alpha+\delta})$ & $\calO(n^{-\alpha+\delta})$ & \citep{dick.higher_order_scrambled_digital_nets, hickernell1998lattice,l2016randomized, dick2022lattice,sloan1994lattice} \\
        digital net (unrandomized) & $\calO(n^{-\alpha+\delta})$ & -- & -- & \citep{niederreiter.qmc_book,dick.digital_nets_sequences_book,dick.higher_order_scrambled_digital_nets,hickernell.generalized_discrepancy_quadrature_error_bound,dick.walsh_spaces_HO_nets,dick.decay_walsh_coefficients_smooth_functions,baldeaux.polylat_efficient_comp_worse_case_error_cbc} \\
        digital net + DS / DP & $\calO(n^{-\alpha+\delta})$ & $\calO(n^{-\alpha+\delta})$ & $\calO(n^{-\alpha+\delta})$ & \citep{niederreiter.qmc_book,dick.digital_nets_sequences_book,hong2003algorithm} \\
        digital net + LMS & $\calO(n^{-\alpha+\delta})$ & $\calO(n^{-\alpha+\delta})$ & $\calO(n^{-\alpha-1/2+\delta})$ & \citep{MATOUSEK1998527,owen.variance_alternative_scrambles_digital_net,tezuka2002randomization} \\
        digital net + NUS & $\calO(n^{-\alpha+\delta})$ & $\calO(n^{-\alpha-1/2+\delta})$ & $\calO(n^{-\alpha-1/2+\delta})$ & \citep{owen1995randomly,dick.digital_nets_sequences_book,l2016randomized,owen.mc_book_practical} \\
        Halton (unrandomized) & $\calO(n^{-1+\delta})$ & -- & -- & \citep{halton1960efficiency,faure2009generalized,lemieux2009monte} \\
        Halton + DS / DP / LMS & $\calO(n^{-1+\delta})$ & $\calO(n^{-1+\delta})$ & $\calO(n^{-1+\delta})$ & \citep{wang2000randomized,owen_halton,qrng.software}
    \end{tabular}  
    }
    \caption{Summary of worst-case error (WCE) and root-mean-squared error (RMSE)
    for various LD point sets and randomizations.}
    \label{tab:wce_rmse_summary}
\end{table}

\Cref{fig:ld_randomizations} depicts many of these randomization routines. In the figure, each vertical line is a unit interval with $0$ at the bottom and $1$ at the top. A given interval is partitioned at the horizontal ticks extending to the right, and then rearranged following the arrows to create the partition of the right interval as shown by the ticks extending to the left. For the shifted rank-$1$ lattice and digital net with a digital shift (DS), when the arrow hits a horizontal gray bar it is wrapped around to the next gray bar below. See for example the dotted line in the digitally shifted digital net. The lattice shift is $\Delta = \thisDelta$. All digital nets use base $b=3$ and $t_\mathrm{max}=2$ digits of precision. The digital shift is $\bDelta = (2,1)^\intercal$. Dropping the $j$ subscript for dimension, the digital permutations in the third panel are $\pi_1 = (2,0,1)$ and $\pi_2 = (2,1,0)$. Notice $\pi_1$ is equivalent to a digital shift by $2$, but $\pi_2$ cannot be written as a digital shift. The nested uniform scramble (NUS) has digital permutations $\pi = (2,0,1)$, $\pi_0 = (2,1,0)$, $\pi_1 = (0,1,2)$, and $\pi_2 = (1,2,0)$. Notice the permuted digital net in the third panel has permutations depending only on the current digit in the base-$b$ expansion. In contrast, the full NUS scrambling in the fourth panel has permutations which depend on all previous digits in the base-$b$ expansion. 

\Cref{tab:wce_rmse_summary} summarizes the worst-case error (WCE) and root-mean-squared error (RMSE) for the unrandomized and randomized LD sequences discussed in this section. In the table $\delta>0$ is an arbitrarily small constant (which may differ from row to row) used to hide logarithmic factors. The parameter $\alpha \geq 1$ is a smoothness parameter with $\alpha>1$ referred to as higher-order convergence, hence the terminology ``higher-order digital nets''. DS stands for digital shifts, DP for digital permutations, LMS for (Matou\v{s}ek's) linear matrix scrambling, and NUS for nested uniform scrambling. 

Lattices are typically studied in weighted periodic Korobov / Sobolev RKHSs of dominating mixed smoothness with the shift invariant (SI) $\alpha$-kernel defined in \Cref{sec:SI_lattices}. For non-periodic integrands in Sobolev spaces, one common strategy is to apply the baker (tent) transform in \eqref{eq:baker} to each coordinate. This periodizes the integrand and places it in a periodic Korobov/Sobolev space with some effective smoothness $\alpha_\mathrm{eff}$. In particular, for integrands with dominating mixed smoothness two, baker-transformed lattice rules can achieve second-order convergence, i.e., WCE and RMSE of order $\calO(n^{-2+\delta})$ \citep{hickernell1998lattice}. In our numerical experiments in \Cref{sec:numerical_experiments} we always apply the baker transform to the non-periodic integrands of study. 

Digital nets are typically studied in either spaces of functions with bounded variation in the sense of Hardy--Krause (BVHK spaces) or in weighted Walsh--Sobolev spaces and the related Walsh--Sobolev RKHSs with the digitally-shift-invariant (DSI) $\alpha$-kernel in \Cref{sec:DSI_dnets}. For digital nets, convergence rates depend on both the smoothness of the function and the design of the digital net; for sufficiently smooth functions, different digital nets (possibly with interlacing) must be used to achieve higher-order convergence. Halton sequences are typically studied in BVHK spaces or analogs of weighted Walsh--Sobolev spaces which account for differing bases per dimension.

For digital nets with LMS, the best known theoretical RMSE is $\calO(n^{-\alpha+\delta})$ whereas empirical RMSEs of order $\calO(n^{-\alpha-1/2+\delta})$ have been observed. In \Cref{sec:numerical_experiments} we replicate this improved empirical RMSE convergence using our \texttt{QMCPy} implementation. We do not list Halton with NUS in \Cref{tab:wce_rmse_summary}, as the corresponding theory is less developed than for digital nets; see \citep{owen.gain_coefficients_scrambled_halton} for some ideas in this direction. 

In practice, we recommend using digital nets (and their higher-order variants) randomized via LMS together with either a digital shift or digital permutation. Empirically, this achieves the same RMSE convergence as NUS, but at a fraction of the cost to generate. This is backed by both our generation-time experiment and RMSE convergence experiment described in \Cref{sec:numerical_experiments}. For periodic functions, lattices are often a better choice. If using lattices for non-periodic functions, the baker transform is always recommended. We recommend randomization for all LD sequences as discussed in \Cref{sec:background}.

\section{Fast transforms and kernel computations} \label{sec:fast_transforms}

\subsection{Overview} 

Recall from the motivation in \Cref{sec:kernel_computations} that we would like to compute the matrix-vector product $\mK \by$ and solve the linear system $\mK^{-1} \by$ where $\by$ is a known vector of length $n$ and $\mK = \{K(\bx_i,\bx_{i'})\}_{i,i'=0}^{n-1}$ is an $n \times n$ symmetric positive definite (SPD) Gram matrix based on a SPD kernel $K$ and point set $\{\bx_i\}_{i=0}^{n-1}$. \Cref{tab:com_kernel_costs} compares our fast kernel methods against standard techniques in terms of both storage and computational costs. Decomposition of the symmetric positive definite (SPD) Gram matrix $\mK$ is the cost of computing the eigendecomposition or Cholesky factorization. The costs of matrix-vector multiplication and solving a linear system are the costs after performing the decomposition. Both storage and kernel computation costs are greatly reduced by pairing certain low-discrepancy point sets with special shift-invariant (SI) or digitally-shift-invariant (DSI) kernels. These fast algorithms rely on the fast Fourier transform in bit-reversed order (FFTBR), its inverse (IFFTBR), and the fast Walsh--Hadamard transform (FWHT). \Cref{fig:fast_transforms} gives schematics of the FFTBR and FWHT algorithms which enable these fast kernel methods. FFTBR is performed via a radix-2 decimation-in-time (DIT) FFT (a base-$2$ Cooley--Tukey algorithm) without the initial reordering of inputs \citep{cooley1965algorithm}. The following subsections detail these special kernel-point set pairings and how the structure induced in $\mK$ can be exploited for fast kernel computations. At the time of writing, our \texttt{QMCPy} implementations only support SI kernels paired to rank-$1$ lattices in natural order and DSI kernels paired to base-$2$ digital nets in natural order. 

\begin{table}[htbp]
    \centering
    \small
    \begin{tabular}{r|lll} 
        $\{\bx_i\}_{i=0}^{n-1}$ & \textbf{any} & \textbf{rank-$1$ lattice} & \textbf{digital net} \\
        $K$ structure & general SPD & SPD SI & SPD DSI \\
        $\mK$ storage & $\calO(n^2)$ & $\calO(n)$ & $\calO(n)$ \\
        $\mK \by$ cost  & $\calO(n^2)$ & $\calO(n \log n)$ & $\calO(n \log n)$ \\
        $\mK^{-1} \by$ cost & $\calO(n^3)$ & $\calO(n \log n)$ & $\calO(n \log n)$ \\
        methods & Cholesky & FFT / IFFT & FWHT 
    \end{tabular}
    \caption{Comparison of storage and cost requirements for kernel methods.}
    \label{tab:com_kernel_costs}
\end{table}


\begin{figure}[htbp]
    \centering
    \newcommand{\h}{1}
    \newcommand{\w}{2}
    \newcommand{\y}{4}
    \newcommand{\z}{6}
    \newcommand{\p}{5}
    \resizebox{\textwidth}{!}{%
    \begin{subfigure}[t]{.45\textwidth}
    \begin{tikzpicture}
        \draw (\z/2,\p+8*\h) node{FFTBR};
        \draw ( 0,\p+7*\h) node[draw,circle](l00){$y_0$}; 
        \draw ( 0,\p+6*\h) node[draw,circle](l10){$y_1$}; 
        \draw ( 0,\p+5*\h) node[draw,circle](l20){$y_2$}; 
        \draw ( 0,\p+4*\h) node[draw,circle](l30){$y_3$}; 
        \draw ( 0,\p+3*\h) node[draw,circle](l40){$y_4$}; 
        \draw ( 0,\p+2*\h) node[draw,circle](l50){$y_5$}; 
        \draw ( 0,\p+1*\h) node[draw,circle](l60){$y_6$}; 
        \draw ( 0,\p+0*\h) node[draw,circle](l70){$y_7$}; 
        \draw (\w,\p+7*\h) node[draw,circle](l01){};
        \draw (\w,\p+6*\h) node[draw,circle](l11){};
        \draw (\w,\p+5*\h) node[draw,circle](l21){};
        \draw (\w,\p+4*\h) node[draw,circle](l31){};
        \draw (\w,\p+3*\h) node[draw,circle](l41){};
        \draw (\w,\p+2*\h) node[draw,circle](l51){};
        \draw (\w,\p+1*\h) node[draw,circle](l61){};
        \draw (\w,\p+0*\h) node[draw,circle](l71){};
        \draw (\y,\p+7*\h) node[draw,circle](l02){};
        \draw (\y,\p+6*\h) node[draw,circle](l12){};
        \draw (\y,\p+5*\h) node[draw,circle](l22){};
        \draw (\y,\p+4*\h) node[draw,circle](l32){};
        \draw (\y,\p+3*\h) node[draw,circle](l42){};
        \draw (\y,\p+2*\h) node[draw,circle](l52){};
        \draw (\y,\p+1*\h) node[draw,circle](l62){};
        \draw (\y,\p+0*\h) node[draw,circle](l72){};
        \draw (\z,\p+7*\h) node[draw,circle](l03){$\ty_0$};
        \draw (\z,\p+6*\h) node[draw,circle](l13){$\ty_1$};
        \draw (\z,\p+5*\h) node[draw,circle](l23){$\ty_2$};
        \draw (\z,\p+4*\h) node[draw,circle](l33){$\ty_3$};
        \draw (\z,\p+3*\h) node[draw,circle](l43){$\ty_4$};
        \draw (\z,\p+2*\h) node[draw,circle](l53){$\ty_5$};
        \draw (\z,\p+1*\h) node[draw,circle](l63){$\ty_6$};
        \draw (\z,\p+0*\h) node[draw,circle](l73){$\ty_7$};
        \draw[->] (l00) -- (l11); 
        \draw[->] (l10) -- (l01); 
        \draw[->] (l20) -- (l31); 
        \draw[->] (l30) -- (l21); 
        \draw[->] (l40) -- (l51); 
        \draw[->] (l50) -- (l41); 
        \draw[->] (l60) -- (l71); 
        \draw[->] (l70) -- (l61);
        \draw (\w/2,\p+0.5*\h) node[draw,circle,fill=black,text=white]{$0$};
        \draw (\w/2,\p+2.5*\h) node[draw,circle,fill=black,text=white]{$0$};
        \draw (\w/2,\p+4.5*\h) node[draw,circle,fill=black,text=white]{$0$};
        \draw (\w/2,\p+6.5*\h) node[draw,circle,fill=black,text=white]{$0$};
        \draw[->] (l01) -- (l22); 
        \draw[->] (l11) -- (l32); 
        \draw[->] (l21) -- (l02); 
        \draw[->] (l31) -- (l12); 
        \draw[->] (l41) -- (l62); 
        \draw[->] (l51) -- (l72); 
        \draw[->] (l61) -- (l42); 
        \draw[->] (l71) -- (l52);
        \draw (\w/2+\y/2,\p+1*\h) node[draw,circle,fill=black,text=white]{$2$};
        \draw (\w/2+\y/2,\p+2*\h) node[draw,circle,fill=black,text=white]{$0$};
        \draw (\w/2+\y/2,\p+5*\h) node[draw,circle,fill=black,text=white]{$2$};
        \draw (\w/2+\y/2,\p+6*\h) node[draw,circle,fill=black,text=white]{$0$};
        \draw[->] (l02) -- (l43); 
        \draw[->] (l12) -- (l53); 
        \draw[->] (l22) -- (l63); 
        \draw[->] (l32) -- (l73); 
        \draw[->] (l42) -- (l03); 
        \draw[->] (l52) -- (l13); 
        \draw[->] (l62) -- (l23); 
        \draw[->] (l72) -- (l33);
        \draw (\y/2+\z/2,\p+2*\h) node[draw,circle,fill=black,text=white]{$3$};
        \draw (\y/2+\z/2,\p+3*\h) node[draw,circle,fill=black,text=white]{$2$};
        \draw (\y/2+\z/2,\p+4*\h) node[draw,circle,fill=black,text=white]{$1$};
        \draw (\y/2+\z/2,\p+5*\h) node[draw,circle,fill=black,text=white]{$0$};
        \draw (0,3*\h) node[draw,circle](wit){$a$};
        \draw (0,0) node[draw,circle](wib){$b$};
        \draw (\y/3+2*\z/3,3*\h) node[draw,circle,minimum size=2cm](wot){$\frac{a+W_n^rb}{\sqrt{2}}$};
        \draw (\y/3+2*\z/3,0) node[draw,circle,minimum size=2cm](wob){$\frac{a-W_n^rb}{\sqrt{2}}$};
        \draw[->] (wit) -- (wob); 
        \draw[->] (wib) -- (wot); 
        \draw (.45*\z,1.5*\h) node[draw,circle,fill=black,text=white]{$r$};
        \draw (.45*\z,3*\h) node{$W_n = e^{-2 \pi \sqrt{-1}/n}$};
    \end{tikzpicture}
    \end{subfigure}
    \hspace{1cm}
    \begin{subfigure}[t]{.45\textwidth}
    \begin{tikzpicture}
        \draw (\z/2,\p+8*\h) node{FWHT};
        \draw ( 0,\p+7*\h) node[draw,circle](l00){$y_0$};
        \draw ( 0,\p+6*\h) node[draw,circle](l10){$y_1$};
        \draw ( 0,\p+5*\h) node[draw,circle](l20){$y_2$};
        \draw ( 0,\p+4*\h) node[draw,circle](l30){$y_3$};
        \draw ( 0,\p+3*\h) node[draw,circle](l40){$y_4$};
        \draw ( 0,\p+2*\h) node[draw,circle](l50){$y_5$};
        \draw ( 0,\p+1*\h) node[draw,circle](l60){$y_6$};
        \draw ( 0,\p+0*\h) node[draw,circle](l70){$y_7$};
        \draw (\w,\p+7*\h) node[draw,circle](l01){};
        \draw (\w,\p+6*\h) node[draw,circle](l11){};
        \draw (\w,\p+5*\h) node[draw,circle](l21){};
        \draw (\w,\p+4*\h) node[draw,circle](l31){};
        \draw (\w,\p+3*\h) node[draw,circle](l41){};
        \draw (\w,\p+2*\h) node[draw,circle](l51){};
        \draw (\w,\p+1*\h) node[draw,circle](l61){};
        \draw (\w,\p+0*\h) node[draw,circle](l71){};
        \draw (\y,\p+7*\h) node[draw,circle](l02){};
        \draw (\y,\p+6*\h) node[draw,circle](l12){};
        \draw (\y,\p+5*\h) node[draw,circle](l22){};
        \draw (\y,\p+4*\h) node[draw,circle](l32){};
        \draw (\y,\p+3*\h) node[draw,circle](l42){};
        \draw (\y,\p+2*\h) node[draw,circle](l52){};
        \draw (\y,\p+1*\h) node[draw,circle](l62){};
        \draw (\y,\p+0*\h) node[draw,circle](l72){};
        \draw (\z,\p+7*\h) node[draw,circle](l03){$\ty_0$};
        \draw (\z,\p+6*\h) node[draw,circle](l13){$\ty_1$};
        \draw (\z,\p+5*\h) node[draw,circle](l23){$\ty_2$};
        \draw (\z,\p+4*\h) node[draw,circle](l33){$\ty_3$};
        \draw (\z,\p+3*\h) node[draw,circle](l43){$\ty_4$};
        \draw (\z,\p+2*\h) node[draw,circle](l53){$\ty_5$};
        \draw (\z,\p+1*\h) node[draw,circle](l63){$\ty_6$};
        \draw (\z,\p+0*\h) node[draw,circle](l73){$\ty_7$};
        \draw[->] (l00) -- (l41); 
        \draw[->] (l10) -- (l51); 
        \draw[->] (l20) -- (l61); 
        \draw[->] (l30) -- (l71); 
        \draw[->] (l40) -- (l01); 
        \draw[->] (l50) -- (l11); 
        \draw[->] (l60) -- (l21); 
        \draw[->] (l70) -- (l31);
        \draw (\w/2,\p+2*\h) node[draw,circle,fill=black,text=white]{};
        \draw (\w/2,\p+3*\h) node[draw,circle,fill=black,text=white]{};
        \draw (\w/2,\p+4*\h) node[draw,circle,fill=black,text=white]{};
        \draw (\w/2,\p+5*\h) node[draw,circle,fill=black,text=white]{};
        \draw[->] (l01) -- (l22); 
        \draw[->] (l11) -- (l32); 
        \draw[->] (l21) -- (l02); 
        \draw[->] (l31) -- (l12); 
        \draw[->] (l41) -- (l62); 
        \draw[->] (l51) -- (l72); 
        \draw[->] (l61) -- (l42); 
        \draw[->] (l71) -- (l52);
        \draw (\w/2+\y/2,\p+1*\h) node[draw,circle,fill=black,text=white]{};
        \draw (\w/2+\y/2,\p+2*\h) node[draw,circle,fill=black,text=white]{};
        \draw (\w/2+\y/2,\p+5*\h) node[draw,circle,fill=black,text=white]{};
        \draw (\w/2+\y/2,\p+6*\h) node[draw,circle,fill=black,text=white]{};
        \draw[->] (l02) -- (l13); 
        \draw[->] (l12) -- (l03); 
        \draw[->] (l22) -- (l33); 
        \draw[->] (l32) -- (l23); 
        \draw[->] (l42) -- (l53); 
        \draw[->] (l52) -- (l43); 
        \draw[->] (l62) -- (l73); 
        \draw[->] (l72) -- (l63);
        \draw (\y/2+\z/2,\p+0.5*\h) node[draw,circle,fill=black,text=white]{};
        \draw (\y/2+\z/2,\p+2.5*\h) node[draw,circle,fill=black,text=white]{};
        \draw (\y/2+\z/2,\p+4.5*\h) node[draw,circle,fill=black,text=white]{};
        \draw (\y/2+\z/2,\p+6.5*\h) node[draw,circle,fill=black,text=white]{};
        \draw (0,3*\h) node[draw,circle](wit){$a$};
        \draw (0,0) node[draw,circle](wib){$b$};
        \draw (\y/3+2*\z/3,3*\h) node[draw,circle,minimum size=2cm](wot){$\frac{a+b}{\sqrt{2}}$};
        \draw (\y/3+2*\z/3,0) node[draw,circle,minimum size=2cm](wob){$\frac{a-b}{\sqrt{2}}$};
        \draw[->] (wit) -- (wob); 
        \draw[->] (wib) -- (wot); 
        \draw (.45*\z,1.5*\h) node[draw,circle,fill=black,text=white]{};
    \end{tikzpicture}
    \end{subfigure}
    }
    \caption{FFT in bit-reversed order (FFTBR) and FWHT schemes.}
    \label{fig:fast_transforms}
\end{figure}  

\subsection{SI kernels, rank-$1$ lattices, and the FFTBF/IFFTBR} \label{sec:SI_lattices}

A kernel is said to be \emph{shift-invariant} (SI) when 
$$K(\bx,\bx') = \tK((\bx-\bx') \bmod 1)$$
for some $\tK$, i.e., the kernel is only a function of the component-wise difference between inputs modulo $1$. One class of SI kernels take the form 
\begin{equation}
    K(\bx,\bx') = \zeta \sum_{\fu \subseteq \{1,\dots,d\}} \gamma_\fu \prod_{j \in \fu} K_{\alpha_j}(x_j, x_j')
    \label{eq:K_SI_full}
\end{equation}
where 
$$K_\alpha(x,x') = \frac{(2\pi)^{2\alpha}}{(-1)^{\alpha+1}(2\alpha)!} B_{2\alpha}((x-x')\bmod 1)$$
with smoothness parameters $\balpha=(\alpha_1,\dots,\alpha_d) \in \bbN^d$, general scale $\zeta>0$, weights $\{\gamma_\fu\}_{\fu \subseteq \{1,\dots,d\}}$, and $B_\ell$ denoting the Bernoulli polynomial of degree $\ell$. The corresponding RKHS $H$ is a weighted periodic unanchored Sobolev space of smoothness $\balpha \in \bbN^d$ with norm 
$$\lVert f \rVert_H^2 := \sum_{\fu \subseteq \{1,\dots,d\}} \frac{1}{(2 \pi)^{2 \lvert \fu \rvert}\gamma_\fu} \int_{[0,1]^{\lvert \fu \rvert}} \llvert \int_{[0,1]^{s - \lvert \fu \rvert}} \left(\prod_{j \in \fu} \frac{\partial^{\alpha_j}}{\partial y_j^{\alpha_j}}\right) f(\by) \D \by_{- \fu} \rrvert^2 \D \by_\fu, $$
where $f:[0,1]^d \to \bbR$, $\by_\fu = (\by_j)_{j \in \fu}$, $\by_{-\fu} := (y_j)_{j \in \{1,\dots,d\} \setminus \fu}$, and $\lvert \fu \rvert$ is the cardinality of $\fu$. The space $H$ is a special case of the weighted Korobov space which has real smoothness parameter $\alpha$ characterizing the rate of decay of Fourier coefficients \citep{kaarnioja.kernel_interpolants_lattice_rkhs,kaarnioja.kernel_interpolants_lattice_rkhs_serendipitous,cools2021fast,cools2020lattice,sloan2001tractability,kuo2004lattice}.

The kernel \eqref{eq:K_SI_full} is the sum over $2^d$ terms and thus becomes impractical to compute for large $d$. A simplified space assumes $\gamma_\fu$ are \emph{product weights} which take the form $\gamma_\fu = \prod_{j \in \fu} \gamma_j$ for some $\{\gamma_j\}_{j=1}^d$. The kernel with product weights takes the simplified form 
\begin{equation}
    K(\bx,\bx') = \zeta \prod_{j=1}^d \left(1+\gamma_j K_{\alpha_j}(x_j,x_j')\right).
    \label{eq:K_SI_prod}
\end{equation}
\citet{kaarnioja.kernel_interpolants_lattice_rkhs_serendipitous} gives a review of lattice-based kernel interpolation in such weighted spaces along with the development of serendipitous weights, which are a special form of high-performance product weights. 

Suppose $n=2^m$ for some $m \in \bbN_0$ and let $\{\bx_i\}_{i=0}^{b^m-1}$ be a shifted lattice with shift $\bDelta \in [0,1)^d$ generated in \emph{linear order}. For $0 \leq i,k < n$
$$K(\bx_i,\bx_k) = K\left(\frac{(i-k) \bg}{n},\bzero\right),$$ 
so the Gram matrix $\mK$ is circulant. Let $R_m(i)$ flip the first $m$ bits of $0 \leq i < 2^m$ in base $b=2$ so that if $i=\sum_{t=0}^{m-1} \mi_t 2^t$ then $R(i) = \sum_{t=0}^{m-1} \mi_{m-t-1} 2^t$. The first step in the radix-2 decimation-in-time (DIT) FFT (a base-$2$ Cooley--Tukey algorithm) \citep{cooley1965algorithm} is to reorder the inputs $\{y_i\}_{i=0}^{2^m-1}$ into bit-reversed order $\{y_{R(i)}\}_{i=0}^{2^m-1}$. The last step in computing the IFFT is to reorder the outputs $\{y_{R(i)}\}_{i=0}^{2^m-1}$ into bit-reversed order $\{y_i\}_{i=0}^{2^m-1}$. Therefore, one should generate lattice points in \emph{natural order} and skip the first step of the FFT and last step of the inverse FFT. We call such algorithms the FFTBR and IFFTBR respectively. 

The Fourier matrix is $\overline{\mF_m} = \{W_m^{ij}\}_{i,j=0}^{2^m-1}$ where $W_m = \exp(-2 \pi \sqrt{-1}/2^m)$. Let $\overline{\mT_m} = \{W_m^{iR_m(j)}\}_{i,j=0}^{2^m-1}$ so that  
$$\overline{\mF_m \{y_i\}_{j=0}^{2^m-1}} = \overline{\mT_m \{y_{R(j)}\}_{j=0}^{2^m-1}}.$$
For $\{\bx_i\}_{i=0}^{2^m-1}$ a lattice in natural order, we have the eigendecomposition
$$\mK = \frac{1}{n} \mT_m \mLambda_m \overline{\mT_m}$$
where $\mLambda_m$ is a diagonal matrix of eigenvalues.

Notice that $\overline{\mT_{m+1}} = \left\{W_{m+1}^{i R_{m+1}(j)}\right\}_{i,j=0}^{2^{m+1}-1}$. For $0 \leq j < 2^m$ we have $R_{m+1}(j) = 2 R_m(j)$ so $W^{i R_{m+1}(j)}_{m+1}= W^{i R_m(j)}_m$. For $2^m \leq j < 2^{m+1}$ we have $R_{m+1}(j) = 2R_m(j-2^m)+1$ so $W^{i R_{m+1}(j)}_{m+1} = W_m^{i R(j-2^m)} W_{m+1}^i$. Moreover, for $0 \leq i < 2^m$ we have $W_{m+1}^{2^m+i} = -W_{m+1}^i$. Define $\tbw_m = \{W_{m+1}^i\}_{i=0}^{2^m-1}$. Then 
$$\overline{\mT_{m+1}} = \begin{bmatrix} \overline{\mT_m} & \diag(\tbw_m) \overline{\mT_m} \\ \overline{\mT_m} & - \diag(\tbw_m) \overline{\mT_m} \end{bmatrix}$$
and $\tbw_{m+1}$ is the $\alpha=2$ interlacing of $\tbw_m$ and $W_{m+1} \tbw_m$.

\subsection{DSI kernels, digital nets, and the FWHT} \label{sec:DSI_dnets}

For $x,y \in [0,1)$ with $x = \sum_{t=1}^{m} \mx_{t-1} b^{-t}$ and $y = \sum_{t=1}^m \my_{t-1} b^{-t}$ let 
$$x \oplus y = \sum_{t=1}^m ((\mx_{t-1} + \my_{t-1}) \bmod b) b^{-t} \quad\text{and}\quad x \ominus y = \sum_{t=1}^m ((\mx_{t-1} - \my_{t-1}) \bmod b) b^{-t}$$
denote digital addition and subtraction respectively. Note that in base $b=2$ these operations are both equivalent to the exclusive or (XOR) between bits.  For vectors $\bx,\by \in [0,1)^d$, digital operations $\bx \oplus \by$ and  $\bx \ominus \by$ act component-wise. 
A kernel is said to be \emph{digitally-shift-invariant} (DSI) when 
$$K(\bx, \bx') = \tK(\bx \ominus \bx')$$
for some $\tK$. 

Suppose $n =b^m$ for some $m \in \bbN_0$ and let $\{\bx_i\}_{i=0}^{b^m-1}$ be a digitally-shifted digital net in \emph{natural order} (possibly of higher-order and/or with LMS). \citet[Theorem 5.3.1 and Theorem 5.3.2]{rathinavel.bayesian_QMC_sobol} shows that, for $b=2$, the Gram matrix $\mK$ is \emph{RSBT} (recursive symmetric block Toeplitz) which implies the eigendecomposition
$$\mK = \frac{1}{n} \mH_m \mLambda_m \mH_m$$ where $\mH_m$ is the $2^m \times 2^m$ Hadamard matrix defined by $\mH^{(0)} = [1]$ and the relationship 
$$\mH_{m+1} = \begin{bmatrix} \mH_m & \mH_m \\ \mH_m & - \mH_m \end{bmatrix}.$$
Multiplying by $\mH_m$ can be done at $\calO(n \log n)$ cost using the FWHT. 
Note that $\mK$ is not generally RSBT when $\{\bx_i\}_{i=0}^{b^m-1}$ has a NUS. 

We now describe one-dimensional DSI kernels $K_\alpha(x,x')$ of higher-order smoothness $\alpha \geq 2$. Weighted sums over products of these one-dimensional kernels may be used to construct higher-dimensional kernels, e.g., those with product weights 
$$K(\bx,\bx') = \zeta \prod_{j=1}^d \left(1+\gamma_j K_{\alpha_j}(x_j,x_j')\right).$$
To begin, let us write $k \in \bbN_0$ as $k = \sum_{\ell=1}^{\#k} \mk_{a_\ell} b^{a_\ell}$ where $a_1 > \cdots > a_{\#k} \geq 0$ are the $\#k$ indices of non-zero digits $\mk_{a_\ell} \in \{1,\dots,b-1\}$ in the base-$b$ expansion of $k$. Then the $k^\mathrm{th}$ Walsh coefficient is 
$$\wal_k(x) = e^{2 \pi \sqrt{-1}/b \sum_{\ell=1}^{\#k} \mk_{a_\ell} \mx_{a_\ell+1}}.$$ 
For $\alpha \in \bbN$ we also define the weight function from \citet{dick.decay_walsh_coefficients_smooth_functions} as
$$\mu_\alpha(k) = \sum_{\ell=1}^{\min\{\alpha,\#k\}} (a_\ell+1).$$
For $\alpha \geq 2$, let the DSI kernel 
\begin{equation}
    K_\alpha(x,x') = \sum_{k \in \bbN} \frac{\wal_k(x \ominus x')}{b^{\mu_\alpha(k)}}
    \label{eq:DSI_kernel}
\end{equation}
have corresponding RKHS $H_\alpha$. For $\alpha \geq 2$, \citep{baldeaux.polylat_efficient_comp_worse_case_error_cbc} showed traditional Sobolev smoothness spaces may be continuously embedded in $H_\alpha$. The case of $\alpha=1$ is treated separately by \citet{dick.multivariate_integraion_sobolev_spaces_digital_nets}. 

The following theorem gives computable forms of a few higher-order DSI kernels in base $b=2$. These kernels are also useful in the computation of the worst-case error (WCE) of QMC with higher-order polynomial-lattices \citep{baldeaux.polylat_efficient_comp_worse_case_error_cbc}. In fact, their paper details how to compute $K_\alpha$ in \eqref{eq:DSI_kernel} at $\calO(\alpha \#x)$ cost where $\#x$ is the number of non-zero digits in the base-$b$ expansion of $x$. Their paper also gives the expressions for $\alpha=2$ and $\alpha=3$ in the theorem below. The $\alpha=4$ expression below is derived in appendix \Cref{appendix:proofs}. While this $\alpha=4$ form contains an infinite sum, the terms in the sum decay rapidly, thus enabling fast numerical evaluation to machine precision. 

\begin{theorem} \label{thm:explicit_DSI_low_order_forms}
    Fix the base $b=2$. Let $\beta(x) = - \lfloor \log_2(x) \rfloor$ and for $\nu \in \bbN$ define $t_\nu(x) = 2^{-\nu \beta(x)}$ where for $x=0$ we set $\beta(x) = t_\nu(x) = 0$. Then \eqref{eq:DSI_kernel} satisfies
    \begin{align*}
        K_2(x) &= -1+-\beta(x) x + \frac{5}{2}\left[1-t_1(x)\right], \\
        K_3(x) &= -1+\beta(x)x^2-5\left[1-t_1(x)\right]x+\frac{43}{18}\left[1-t_2(x)\right], \\
        K_4(x) &= -1+- \frac{2}{3}\beta(x)x^3+5\left[1-t_1(x)\right]x^2 - \frac{43}{9}\left[1-t_2(x)\right]x \\
        &+\frac{701}{294}\left[1-t_3(x)\right]+\beta(x)\left[\frac{1}{48}\sum_{a=0}^\infty \frac{\wal_{2^a}(x)}{2^{3a}} - \frac{1}{42}\right].
    \end{align*}
\end{theorem}


\subsection{Fast kernel computations and efficient eigenvalue updates} \label{sec:efficient_updates}

The formulations in the previous two subsections have the following properties. 
\begin{enumerate}
    \item For $n=2^m$, the Gram matrix $\mK_m \in \bbR^{n \times n}$ can be written as
    $$\mK_m = \mV_m \mLambda_m \overline{\mV_m}$$
    where $\mV_m \overline{\mV_m} = \mI$ and the first column of $\mV_m$, denoted $\bv_{m1}$, is the constant $1/\sqrt{n}$.
    \item $\mV_m \in \bbC^{n \times n}$ satisfies $\mV_0 = [1]$ and 
    \begin{equation}
        \mV_{m+1} = \begin{bmatrix} \mV_m & \mV_m \\  \mV_m \overline{\diag(\tbw_m)} & - \mV_m \overline{\diag(\tbw_m)} \end{bmatrix}/\sqrt{2}
        \label{eq:ft_next_parts}
    \end{equation}
    with $\tbw_m \in \bbC^n$ soon to be defined. 
    \item $\mV_m \by$ and $\overline{\mV_m} \by$ are each computable at cost $\calO(2^mm)= \calO(n \log n)$.
\end{enumerate}
For the case of lattice points with SI kernels, $\mV_m = \mT_m/\sqrt{2^m}$ is the scaled and permuted Fourier matrix and $\tbw_m = \{\exp(-\pi \sqrt{-1}/2^m i)\}_{i=0}^{2^m-1}$, so $\overline{\mV_m} \by$ and $\mV_m \by$ can be computed using the FFTBR and IFFTBR respectively. For the case of digital nets with 
DSI kernels, $\mV_m = \mH_m/\sqrt{n}$ is the real symmetric Hadamard matrix and $\tbw_m = \{1\}_{i=0}^{2^m-1}$, so both $\overline{\mV_m} \by$ and $\mV_m \by$ can be computed using a FWHT. 

With $\mLambda_m = \diag(\blambda_m)$, we have 
$$\blambda_m = \sqrt{n} \mLambda_m\overline{\bv_{m1}} = \sqrt{n}\;\overline{\mV_m} \left(\mV_m \mLambda_m \overline{\bv_{m1}}\right) = \sqrt{n} \; \overline{\mV_m} \bk_{m1}$$
which can be computed at $\calO(n \log n)$ cost and only requires storing the first column of $\mK^{(m)}$ which we denote by $\bk_{m1}$. Moreover, 
\begin{equation*}
    \mK \by = \mV_m(\tby \odot \blambda_m) \qquad \text{and}\qquad \mK^{-1} \by = \mV_m (\tby \odiv \blambda_m)
\end{equation*}
may each be evaluated at cost $\calO(n \log n)$ where $\odot$ denotes the Hadamard (component-wise) product and $\odiv$ denotes component-wise division. 

The following code implements the fast Gram matrix multiplication $\bu := \mK \by$ and the fast Gram matrix linear system solution $\bv := \mK^{-1} \by$ for the SI-lattice variant. It also includes an application of \eqref{eq:ft_next_parts} to update transformed values, i.e., to compute $\tby_\mathrm{full} = \overline{\mV_{m+1}} \begin{bmatrix} \by \\ \by_\mathrm{new} \end{bmatrix}$ given $\tby = \overline{\mV_m} \by$. An analogous code for the DSI-digital net variant is given in Listing \ref{code:DSI_dnet}.

\begin{lstlisting}[style=Python, caption={Efficient Gram matrix operations and transform updates for the SI-lattice variant.}, label={code:SI_lattice}]
d = 3 # dimension 
m = 10 # will generate 2^m points
n = 2**m # number of points
lattice = qp.Lattice(d) # defaults to natural order
kernel = qp.KernelShiftInvar(
    d, # dimension 
    alpha = [1,2,3], # per-dim smoothness parameters
    weights = [1,1/2,1/4]) # per-dim product weights
x = lattice(n_min=0,n_max=n) # shape=(n, d) lattice
y = np.random.rand(n) # shape=(n,) random uniforms
# fast matrix multiplication and linear system solve
k1 = kernel(x,x[0]) # shape=(n,) first col of Gram matrix
lam = np.sqrt(n)*qp.fftbr(k1) # vector of eigenvalues
yt = qp.fftbr(y)
u = qp.ifftbr(yt*lam) # fast matrix multiplication 
v = qp.ifftbr(yt/lam) # fast linear system solve
# efficient fast transform updates
ynew = np.random.rand(n) # shape=(n,) new random uniforms
omega = qp.omega_fftbr(m) # shape=(n,)
ytnew = qp.fftbr(ynew) # shape=(n,)
ytfull = 1/np.sqrt(2)*np.hstack([ # ytfull shape=(2n,)
        yt+omega*ytnew, #  shape=(n,) first half
        yt-omega*ytnew]) # shape=(n,) second half
\end{lstlisting}


\section{Numerical experiments} \label{sec:numerical_experiments}

All numerical experiments were carried out on a 2025 M4 MacBook Air. \Cref{fig:timing} compares the wall-clock time required to generate point sets and perform fast transforms for high dimensions and/or number of randomizations. For IID and randomized LD generators, we vary the number of randomizations $R$, the number of dimensions $d$, and the sequence size $n$. Timings include initialization, and a new point set is generating for each plotted $(R,n,d)$. Fast transforms are applied to $R$ sequences of size $n$ in a vectorized fashion. 

Following the theory, our implementation scales linearly in the number of dimensions and randomizations. For large numbers of randomizations, our vectorized \texttt{QMCPy} LD generators significantly outperform the looped implementations in \texttt{SciPy} and \texttt{PyTorch}. For fast transforms, our FFTBR and IFFTBR implementations are the same speed as \texttt{SciPy}'s FFT and IFFT algorithms. Our FWHT implementation is significantly faster than the \texttt{SymPy} implementation, especially when applying the FWHT to multiple sequences simultaneously. IID points, randomly shifted lattices, and digital nets (DNs) with LMS and digital shifts (DSs), including higher-order versions, are the fastest sequences to generate. Digital nets in base $b=2$ exploit Gray code order, integer storage of bit-vectors, and exclusive or (XOR) operations to perform digital addition. Halton point sets are slower to generate as they cannot exploit these advantages. NUS, especially higher-order versions, are significantly slower to generate than LMS randomizations. Even so, an LMS interlacing scrambling of digital shifts empirically achieves the $\calO(n^{-\alpha-1/2+\delta})$ RMSE convergence rate which is theoretically guaranteed only for NUS scrambling, as we show in the next experiment.  

\begin{figure}[htbp]
    \centering
    \includegraphics[width=1\textwidth]{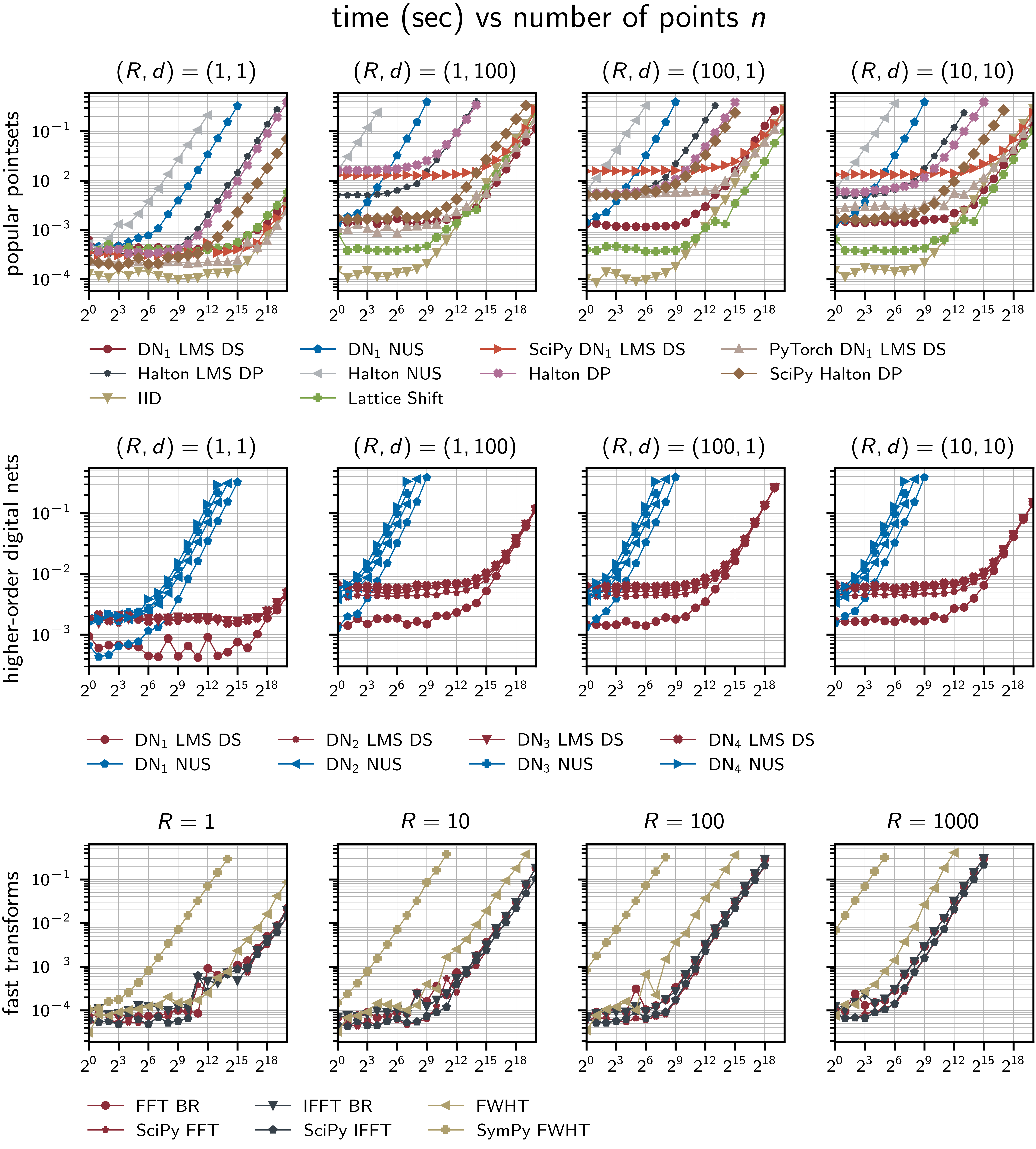}
    \caption{Comparison of time required to generate point sets and perform fast transforms.}
    \label{fig:timing}
\end{figure}

\Cref{fig:convergence} plots the RMSE convergence of RQMC methods applied to the integrands described below.
\begin{description}
    \item[Simple function, $d=1$,] has $f(x) = x e^x-1$. This was used by \citet{dick.higher_order_scrambled_digital_nets} where higher-order digital net scrambling was first proposed.
    \item[Simple function, $d=2$,] has $f(\bx) = x_2 e^{x_1 x_2}/(e-2)-1$. This was also considered by \citet{dick.higher_order_scrambled_digital_nets}. 
    \item[Oakley \& O'Hagan, $d=2$,] has $f(\bx) = g((\bx-1/2)/50)$ for $g(\bt) = 5+t_1+t_2+2\cos(t_1)+2\cos(t_2)$, see the original work of \citet{oakley2002bayesian} or the VLSE (Virtual Library of Simulation Experiments, \url{https://www.sfu.ca/~ssurjano/uq.html}).
    \item[G-Function, $d=3$,] has periodic $f(\bx) = \prod_{j=1}^d \frac{\lvert 4x_j-2\rvert-a_j}{1+a_j}$ with $a_j = (j-2)/2$ for $1 \leq j \leq d$, see the works of \citet{crestaux2007polynomial} and \citet{marrel2009calculations} or the VLSE. 
    \item[Oscillatory Genz, $d=3$,] has $f(\bx) = \cos\left(-\sum_{j=1}^d c_j x_j \right)$ with coefficients of the third kind $c_j = 4.5 \tc_j/\sum_{j=1}^d \tc_j$ where $\tc_j = \exp\left(j \log\left(10^{-8}\right)/d\right)$. This is a common test function for numerical integration quantification which is available in the \texttt{Dakota} software \citep{adams2020dakota} among others.  
    \item[Corner-peak Genz, $d=3$,] has $f(\bx) = \left(1+\sum_{j=1}^d c_j x_j\right)^{-(d+1)}$ with coefficients of the second kind $c_j = 0.25 \tc_j/\sum_{j=1}^d \tc_j$ where $\tc_j = 1/j^2$. This is also available in \texttt{Dakota} for testing numerical integration.
\end{description}
For each problem, the RMSE of the (Q)MC estimator $\hmu$ in \eqref{eq:mc_approx} is approximated using $300$ independent randomizations of an IID or randomized LD point sets from \texttt{QMCPy}. IID points consistently achieve the theoretical $\calO(n^{-1/2})$ convergence rate. For shifted lattices, we periodized the integrand using a baker transform \cref{eq:baker} which does not change the mean, i.e., we use $\tilde{f}(\bx) = f(1-2\lvert \bx-1/2\rvert)$ in place of $f(\bx)$ in \eqref{eq:mc_approx}. Shifted lattices consistently attained RMSEs like $\calO(n^{-2})$, even for non-periodic functions, as expected from \citep{hickernell1998lattice}. Our base-$2$ digital nets with LMS and DS, including higher-order versions with higher-order LMS, consistently achieved the lowest RMSEs and best rates of convergence. For the $d=1$ integrand, we are able to realize a rate of $\calO(n^{-\min\{\alpha,\talpha\}-1/2-\delta})$ where $\alpha$ is the higher-order digital interlacing of the net, $\talpha$ is the smoothness of the integrand, and $\delta>0$ is arbitrarily small. For the $2$-dimensional and $3$-dimensional integrands we were also able to achieve higher-order convergence in some cases, but are hindered by the limited smoothness and anisotropy of the test functions, as well as increasing dimensional effects that make the optimal rates harder to realize uniformly across all problems \citep{dick.high_dim_integration_qmc_way}. These empirical RMSE rates for LMS with DS match the theoretical rates for NUS, while requiring only a fraction of the generation time (see \Cref{fig:timing}), thus supporting \Cref{tab:wce_rmse_summary} and the practical recommendations in \Cref{sec:summary_convergence_practical_guidance}.

\begin{figure}[htbp]
    \centering
    \includegraphics[width=1\textwidth]{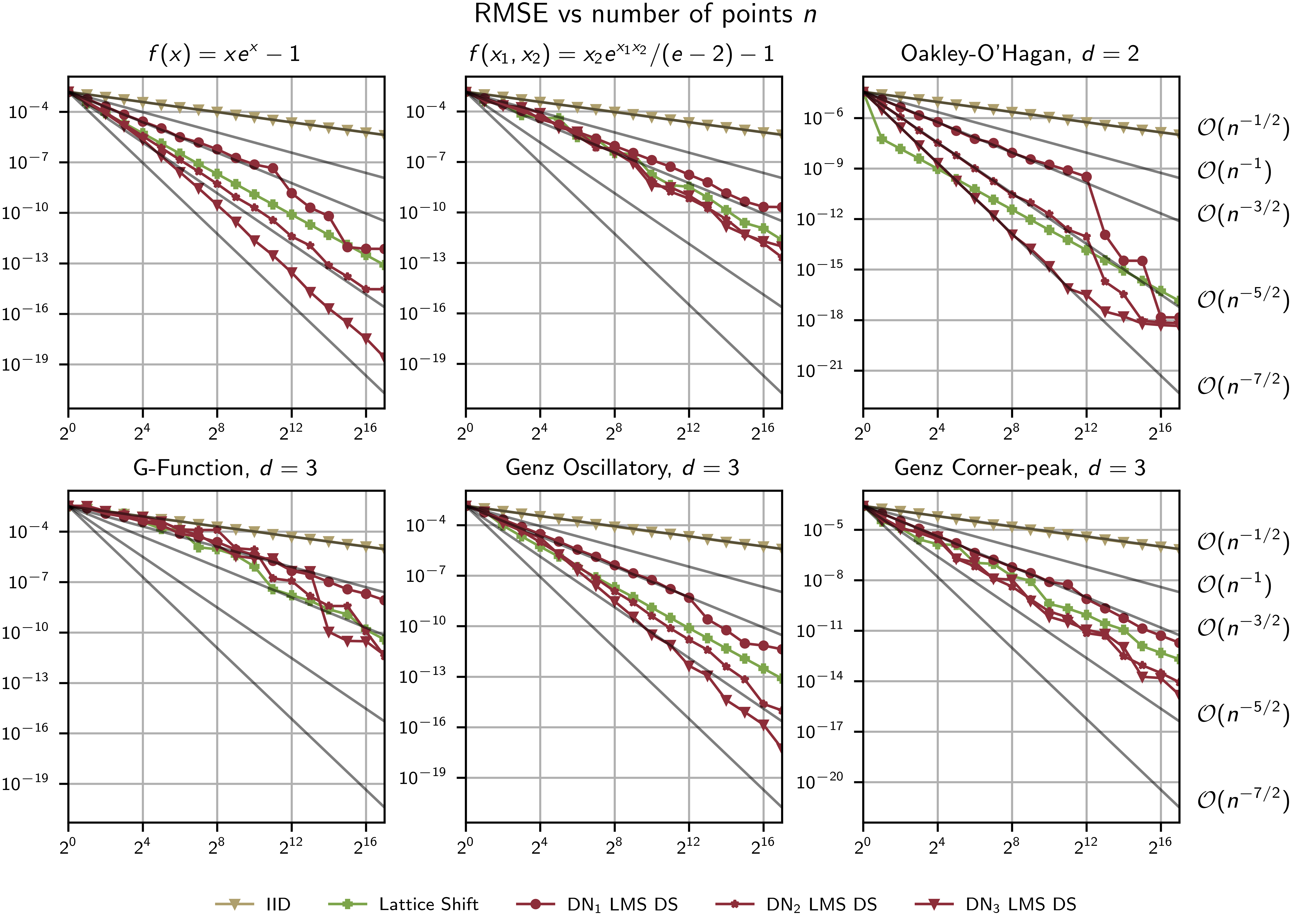}
    \caption{The RMSE of the RQMC estimate for a few different integrands.}
    \label{fig:convergence}
\end{figure}

In the following code, we use \texttt{QMCPy} to build an RQMC estimate to the mean of the Corner-peak Genz function in $d=50$ dimensions. Using $R$ replications and a fixed number of points $n$, we construct a Student's-$t$ confidence interval as in \eqref{eq:RQMC} and the discussion thereafter. \texttt{SciPy} is used to compute the necessary quantile. Listing \ref{code:Genz_ex_2} gives an example of an adaptive version of this algorithm which automatically selects the number of points $n$ required to meet a user-specified absolute error tolerance.

\begin{lstlisting}[style=Python, caption={Randomized QMC approximation and error estimation for a fixed number of points.}, label={code:Genz_ex_1}]
import scipy.stats
def gen_corner_peak_2(x):
    # x.shape=(...,n,d), e.g., (n,d) or (R,n,d) 
    d = x.shape[-1]
    c_tilde = 1/np.arange(1,d+1)**2
    c = 0.25*c_tilde/np.sum(c_tilde)
    y = (1+np.sum(c*x,axis=-1))**(-(d+1)) 
    return y # y.shape=(...,n), e.g., (n,) or (R,n)
R = 10 # number of randomizations
n = 2**15 # number of points 
d = 50 # dimension
dnb2 = qp.DigitalNet(d, replications=R, seed=7, alpha=3)
x = dnb2(n) # x.shape=(R,n,d)
y = gen_corner_peak_2(x) # y.shape=(R,n) 
muhats_r = np.mean(y,axis=1) # muhats_r.shape=(R,)
muhat = np.mean(muhats_r) # muhat is a scalar 
print(muhat)
""" 0.014936813948394042 """
alpha = 0.01 # uncertainty level
t_star = -scipy.stats.t.ppf(alpha/2,df=R-1) # quantile
sigmahat = np.std(muhats_r,ddof=1) # unbiased estimate
std_error = t_star*sigmahat/np.sqrt(R)
print(std_error)
""" 5.247445301861484e-07 """
conf_int = [muhat-std_error,muhat+std_error]
\end{lstlisting}

\section{Conclusion} \label{sec:conclusion}

This work has reviewed our  \texttt{QMCPy} implementations of routines for generating randomized LD sequences and applying them to fast kernel methods. We presented lattices, digital nets, and Halton point sets with randomizations spanning shifts, digital permutations, digital shifts, LMS, and NUS. Higher-order scramblings for digital nets were also considered, including higher-order NUS and LMS with interlacing. We described classes of kernels which pair with LD sequences to enable fast kernel computations. These included SI and DSI kernels of higher-order smoothness. The fast kernel methods utilized the fast transform algorithms FFTBR, IFFTBR, and FWHT. We showed our implementations in \texttt{QMCPy} achieve state-of-the-art speed, accuracy, and flexibility.

\section*{Acknowledgements}
    Thank you to Fred J. Hickernell, Sou-Cheng T. Choi, and Aadit Jain for helpful comments and prototyping. Thank you to the reviewers for helpful feedback and suggestions. 

\bibliographystyle{elsarticle-harv} 
\bibliography{main}

\appendix 

\section{Extra code snippets} \label{appendix:codes}

\begin{lstlisting}[style=Python, caption={Efficient Gram matrix operations and transform updates for the DSI-digital net variant.}, label={code:DSI_dnet}]
d = 3 # dimension 
m = 10 # will generate 2^m points
n = 2**m # number of points
dnb2 = qp.DigitalNetB2(d) # defaults to natural order
kernel = qp.KernelDigShiftInvar(
    d, # dimension 
    t = dnb2.t, # bits in int representation of points
    alpha = [1,2,3], # per-dim smoothness parameters
    weights = [1,1/2,1/4]) # per-dim product weights
x = dnb2(n_min=0,n_max=n) # shape=(n, d) digital net
y = np.random.rand(n) # shape=(n,) random uniforms
# fast matrix multiplication and linear system solve
k1 = kernel(x,x[0]) # shape=(n,) first col of Gram matrix
lam = np.sqrt(n)*qp.fwht(k1) # vector of eigenvalues
yt = qp.fwht(y)
u = qp.fwht(yt*lam) # fast matrix multiplication 
v = qp.fwht(yt/lam) # fast linear system solve
# efficient fast transform updates
ynew = np.random.rand(n) # shape=(n,) new random uniforms
omega = qp.omega_fwht(m) # shape=(n,)
ytnew = qp.fwht(ynew) # shape=(n,)
ytfull = 1/np.sqrt(2)*np.hstack([ # ytfull shape=(2n,)
        yt+omega*ytnew, #  shape=(n,) first half
        yt-omega*ytnew]) # shape=(n,) second half
\end{lstlisting}

\begin{lstlisting}[style=Python, caption={Adaptive randomized QMC approximation and error estimation.}, label={code:Genz_ex_2}]
def gen_corner_peak_2(x):
    # x.shape=(...,n,d), e.g., (n,d) or (R,n,d) 
    d = x.shape[-1] 
    c_tilde = 1/np.arange(1,d+1)**2
    c = 0.25*c_tilde/np.sum(c_tilde)
    y = (1+np.sum(c*x,axis=-1))**(-(d+1)) 
    return y # y.shape=(...,n), e.g., (n,) or (R,n)
R = 10
d = 50 
dnb2 = qp.DigitalNet(d, replications=R, seed=7, alpha=3)
true_measure = qp.Uniform(
    sampler = dnb2,
    lower_bound = 0,
    upper_bound = 1)
integrand = qp.CustomFun(
    true_measure = true_measure,
    g = gen_corner_peak_2)
# equivalent to 
# integrand = qp.Genz(
#   sampler = dnb2,
#   kind_func = "CORNER PEAK",
#   kind_coeff = 2)
qmc_algo = qp.CubQMCRepStudentT(integrand, abs_tol=1e-4)
solution,data = qmc_algo.integrate() # run adaptive QMC 
print(solution)
""" 0.014950908095474802 """
conf_int = [data.comb_bound_low,data.comb_bound_high]
std_error = (conf_int[1]-conf_int[0])/2
print(std_error)
""" 2.7968149935497788e-05 """
print(data)
"""
Data (Data)
    solution        0.015
    comb_bound_low  0.015
    comb_bound_high 0.015
    comb_bound_diff 5.59e-05
    n               10240
    n_rep           2^(10)
    time_integrate  0.019
CubQMCRepStudentT (AbstractStoppingCriterion)
    alpha           0.010
    abs_tol         1.00e-04
    rel_tol         0
    n_init          2^(8)
DigitalNetB2 (AbstractLDDiscreteDistribution)
    d               50
    replications    10
    randomize       LMS DS
    gen_mats_source joe_kuo.6.21201.txt
    order           RADICAL INVERSE
    t               63
    alpha           3
    entropy         7
"""
\end{lstlisting}

\section{Proofs of theorems} \label{appendix:proofs}

For $k \in \bbN_0$ write 
$$\hf(k) := \int_0^1 f(x) \overline{\wal_k(x)} \D x.$$

\begin{lemma}[Walsh coefficients of low order monomials]\label{lemma:walsh_low_order_monomials}
    Fix $b=2$.  Let $f_p(x) := x^p$. When $k \in \bbN$ write 
    $$k = 2^{a_1}+\dots+2^{a_{\#k}}$$
    where $a_1 > a_2 > \dots > a_{\#k} \geq 0$. Then we have 
    \begin{align*}
        \hf_1(k) &= \begin{cases} 1/2, & k = 0 \\ -2^{-a_1-2}, & k=2^{a_1} \\ 0, & \mathrm{otherwise} \end{cases}, \\
        \hf_2(k) &= \begin{cases} 1/3, & k = 0 \\ -2^{-a_1-2}, & k=2^{a_1} \\ 2^{-a_1-a_2-3}, & k = 2^{a_1}+2^{a_2} \\ 0, & \mathrm{otherwise} \end{cases}, \\
        \hf_3(k) &= \begin{cases} 1/4, & k=0 \\ -2^{-a_1-2} + 2^{-3a_1-5}, & k = 2^{a_1} \\ 3 * 2^{-a_1-a_2-4}, & k=2^{a_1}+2^{a_2} \\ -3 * 2^{-a_1-a_2-a_3-5}, & k=2^{a_1}+2^{a_2}+2^{a_3} \\ 0, & \mathrm{otherwise} \end{cases}.
    \end{align*}
\end{lemma}
\begin{proof}
    The forms for $\hf_1$ and $\hf_2$ were derived by \citet[Example 14.2, Example 14.3]{dick.digital_nets_sequences_book}. 
    For $k=0$ and any $x \in [0,1)$ we have $\wal_0(x) = 1$, so  
    $$\hf_3(x) = \int_0^1 x^3 \D x = 1/4.$$ 
    Assume $k \in \bbN$ going forward. For $k=2^{a_1}+k'$ with $0 \leq k' < 2^{a_1}$, \citet[Equation 3.6]{fine.walsh_functions} implies
    $$J_k(x) := \int_0^x \wal_k(t) \D t = 2^{-a_1-2} \left[\wal_{k'}(x) - \sum_{r=1}^\infty 2^{-r} \wal_{2^{a_1+r}+k}(x)\right].$$
    Using integration by parts and the fact that $J_k(0) = J_k(1) = 0$,
    \begin{align*}
        \hf_3(k) &= \int_0^1 x^3 \wal_k(x) \D x 
        = \left[x^3 J_k(x) \right]_{x=0}^{x=1} - 3 \int_0^1 x^2 J_k(x) \D x \\
        &= -3*2^{-a_1-2} \left[\hf_2(k') - \sum_{r=1}^\infty 2^{-r} \hf_2(2^{a_1+r}+k)\right].
    \end{align*}
    \begin{itemize}
        \item If $\#k=1$, i.e., $k=2^{a_1}$ then 
        \begin{align*}
            \hf_3(k) &= -3*2^{-a_1-2} \left[\hf_2(0) - \sum_{r=1}^\infty 2^{-r} \hf_2(2^{a_1+r}+2^{a_1})\right] \\
            &= -3*2^{-a_1-2} \left[\frac{1}{3} - \sum_{r=1}^\infty 2^{-r} 2^{-(a_1+r)-a_1-3}\right] \\
            &= 2^{-3a_1-5} - 2^{-a_1-2}.
        \end{align*}
        \item If $\#k=2$ then 
        $$\hf_3(k) = -3*2^{-a_1-2} \hf_2(2^{a_2}) = 3 * 2^{-a_1-a_2-4}.$$
        \item If $\#k=3$ then 
        $$\hf_3(k) = -3*2^{-a_1-2} \hf_2(2^{a_2}+2^{a_3}) = -3 * 2^{-a_1-a_2-a_3-5}.$$
        \item If $\#k>3$ then $\hf_3(k)=0$.
    \end{itemize}
\end{proof}

\begin{proof}[Proof of \Cref{thm:explicit_DSI_low_order_forms}]
    Write  
    $$K_\alpha(x) = \sum_{1 \leq \nu < \alpha} s_\nu(x) + \ts_\alpha(x)$$  
    where $s_\nu$ sums over all $k$ with $\#k = \nu$ and $\ts_\alpha$ sums over all $k$ with $\#k \geq \alpha$. \citet[Corollary 1]{baldeaux.polylat_efficient_comp_worse_case_error_cbc} showed that 
    \begin{align*}
        s_1(x) &= -2x+1, \\
        s_2(x) &= 2x^2 - 2x + \frac{1}{3}, \\
        \ts_2(x) &=  \left[2-\beta(x)\right]x + \frac{1}{2}\left[1-5t_1(x)\right], \\
        \ts_3(x) &= -\left[2-\beta(x)\right] x^2 - \left[1-5t_1(x)\right]x + \frac{1}{18}\left[1-43t_2(x)\right]
    \end{align*}
    from which $K_2$ and $K_3$ follow. We now find expressions for $s_3$ and $\ts_4$ from which $K_4$ follows. 

    Assume sums over $a_i$ are over $\bbN_0$ unless otherwise restricted. \Cref{lemma:walsh_low_order_monomials} gives
    \begin{align*}
        x &= \frac{1}{2} - \sum_{a_1} \frac{\wal_{2^{a_1}}(x)}{2^{a_1+2}}, \\
        x^2 &= \frac{1}{3} - \sum_{a_1}\frac{\wal_{2^{a_1}}(x)}{2^{a_1+2}} + \sum_{a_1>a_2} \frac{\wal_{2^{a_1}+2^{a_2}}(x)}{2^{a_1+a_2+3}}, \\
        x^3 &= \frac{1}{4} - \sum_{a_1} \frac{\wal_{2^{a_1}}(x)}{2^{a_1+2}} + \frac{3}{2}\sum_{a_1>a_2} \frac{\wal_{2^{a_1}+2^{a_2}}(x)}{2^{a_1+a_2+3}} - \frac{3}{2} \sum_{a_1>a_2>a_3} \frac{\wal_{2^{a_1}+2^{a_2}+2^{a_3}}(x)}{2^{a_1+a_2+a_3+4}} + \sum_{a_1} \frac{\wal_{2^{a_1}}(x)}{2^{3a_1+5}}
    \end{align*}
    so
    $$x^3-\frac{3}{2}x^2+\frac{1}{2}x = \frac{1}{32} \sum_{a_1} \frac{\wal_{2^{a_1}}(x)}{2^{3a_1}} - \frac{3}{4} \sum_{a_1>a_2>a_3} \frac{\wal_{2^{a_1}+2^{a_2}+2^{a_3}}(x)}{2^{a_1+a_2+a_3+3}}$$
    and 
    $$s_3(x) = \sum_{a_1>a_2>a_3} \frac{\wal_{2^{a_1}+2^{a_2}+2^{a_3}}(x)}{2^{a_1+a_2+a_3+3}} = -\frac{4}{3} x^3+ 2x^2 -\frac{2}{3}x + \frac{1}{24} \sum_{a_1} \frac{\wal_{2^{a_1}}(x)}{2^{3a_1}}.$$
    Now,
    \begin{align*}
        \ts_4(x) &= \sum_{\substack{a_1>a_2>a_3>a_4 \\ 0 \leq k < 2^{a_4}}} \frac{\wal_{2^{a_1}+2^{a_2}+2^{a_3}+2^{a_4}+k}(x)}{2^{a_1+a_2+a_3+a_4+4}} \\
        &= \sum_{a_1>a_2>a_3>a_4} \frac{\wal_{2^{a_1}+2^{a_2}+2^{a_3}+2^{a_4}}(x)}{2^{a_1+a_2+a_3+a_4+4}} \sum_{0 \leq k < 2^{a_4}} \wal_k(x).
    \end{align*}
    If $x=0$, then 
    $$\ts_4(0) = \sum_{a_1 > a_2 > a_3 > a_4} \frac{1}{2^{a_1+a_2+a_3+4}} = \frac{1}{294}.$$ 
    Going forward, assume $x \in (0,1)$ so $\beta(x) = - \lfloor \log_2(x) \rfloor$ is finite. Recall that  
    $$\sum_{0 \leq k < 2^{a_4}} \wal_k(x) = \begin{cases} 2^{a_4}, & a_4 \leq \beta(x)-1 \\ 0, & a_4 > \beta(x)-1 \end{cases}.$$
    Moreover, since $\beta(x)$ is the index of the first $1$ in the base-$2$ expansion of $x$, when $a_4 < \beta(x)-1$ we have $\wal_{2^{a_4}}(x) = (-1)^{\mx_{a_4+1}} = 1$ and when $a_4 = \beta(x)-1$ we have $\wal_{2^{a_4}}(x) = -1$. This implies 
    \begin{align*}
        \ts_4(x) &= \sum_{\substack{a_1>a_2>a_3>a_4 \\ \beta(x)-1 \geq a_4}} \frac{\wal_{2^{a_1}+2^{a_2}+2^{a_3}+2^{a_4}}(x)}{2^{a_1+a_2+a_3+4}} \\
        &= \sum_{\substack{a_1>a_2>a_3>a_4 \\ \beta(x)-1 > a_4}} \frac{\wal_{2^{a_1}+2^{a_2}+2^{a_3}}(x)}{2^{a_1+a_2+a_3+4}} - \sum_{a_1>a_2 > a_3 > \beta(x)-1} \frac{\wal_{2^{a_1}+2^{a_2}+2^{a_3}}(x)}{2^{a_1+a_2+a_3+4}} \\
        &=: T_1 - T_2.
    \end{align*}
    
    The first term is 
    \begin{align*}
        T_1 &= \sum_{\beta(x)-1 > a_4} \bigg(\sum_{a_1>a_2>a_3} \frac{\wal_{2^{a_1}+2^{a_2}+2^{a_3}}(x)}{2^{a_1+a_2+a_3+4}} -  \sum_{a_4  \geq a_1>a_2>a_3} \frac{1}{2^{a_1+a_2+a_3+4}} \\ & \qquad\qquad  - \sum_{a_1 > a_4 \geq a_2 > a_3} \frac{\wal_{2^{a_1}}(x)}{2^{a_1+a_2+a_3+4}} - \sum_{a_1 > a_2 > a_4 \geq a_3} \frac{\wal_{2^{a_1}+2^{a_2}}(x)}{2^{a_1+a_2+a_3+4}} 
        \bigg) \\
        &=: \sum_{\beta(x)-1 > a_4} \left[V_1(a_4)-V_2(a_4) - V_3(a_4) - V_4(a_4) \right].
    \end{align*}
    Clearly $V_1(a_4) = s_3(x)/2$ and $V_2$ is easily computed. Now
    \begin{align*}
        V_3(a_4) &= \left(\sum_{a_4 \geq a_2 > a_3} \frac{1}{2^{a_2+a_3+3}}\right)\left(\sum_{a_1>a_4} \frac{\wal_{2^{a_1}}(x)}{2^{a_1+1}}\right) \\
        &= \left(\sum_{a_4 \geq a_2 > a_3} \frac{1}{2^{a_2+a_3+3}}\right)\left(s_1(x) - \sum_{a_4 \geq a_1} \frac{1}{2^{a_1+1}} \right)
    \end{align*}
    and 
    \begin{align*}
        V_4(a_4) & = \left(\sum_{a_4 \geq a_3} \frac{1}{2^{a_3+2}}\right)\left(\sum_{a_1>a_2 > a_4} \frac{\wal_{2^{a_1}+2^{a_2}}(x)}{2^{a_1+a_2+2}}\right) \\
        &= \left(\sum_{a_4 \geq a_3} \frac{1}{2^{a_3+2}}\right)\left(\sum_{a_1>a_2} \frac{\wal_{2^{a_1}+2^{a_2}}(x)}{2^{a_1+a_2+2}} - \sum_{a_4 \geq a_1 > a_2} \frac{1}{2^{a_1+a_2+2}} - \sum_{a_1 > a_4 \geq a_2} \frac{\wal_{2^{a_1}}(x)}{2^{a_1+a_2+2}}\right) \\
        &= \left(\sum_{a_4 \geq a_3} \frac{1}{2^{a_3+2}}\right)\left(s_2(x) - \sum_{a_4 \geq a_1 > a_2} \frac{1}{2^{a_1+a_2+2}} - \left(s_1(x) - \sum_{a_4 \geq a_1} \frac{1}{2^{a_1+1}} \right)\left(\sum_{a_4 \geq a_2} \frac{1}{2^{a_2+1}}\right)\right).
    \end{align*}
    
    The second term is 
    \begin{align*}
        T_2 &= \sum_{a_1>a_2>a_3} \frac{\wal_{2^{a_1}+2^{a_2}+2^{a_3}}(x)}{2^{a_1+a_2+a_3+4}} - \sum_{\beta(x)-1 > a_1 > a_2 > a_3} \frac{1}{2^{a_1+a_2+a_3+4}} + \sum_{\beta(x)-1 > a_2 > a_3} \frac{1}{2^{\beta(x)+a_2+a_3+3}} \\ & - \sum_{a_1 > \beta(x)-1 > a_2 > a_3} \frac{\wal_{2^{a_1}}(x)}{2^{a_1+a_2+a_3+4}} + \sum_{a_1 > \beta(x)-1 > a_3}  \frac{\wal_{2^{a_1}}(x)}{2^{a_1+\beta(x)+a_3+3}} \\ &- \sum_{a_1 > a_2 > \beta(x)-1 > a_3}\frac{\wal_{2^{a_1}+2^{a_2}}(x)}{2^{a_1+a_2+a_3+4}} + \sum_{a_1 > a_2 > \beta(x)-1} \frac{\wal_{2^{a_1}+2^{a_2}}(x)}{2^{a_1+a_2+\beta(x)+3}} \\
        &=: W_1-W_2+W_3-W_4+W_5-W_6+W_7.
    \end{align*}
    Clearly $W_1 = s_3(x)/2$ and both $W_2$ and $W_3$ are easily computed. Similarity in the next two sums gives
    \begin{align*}
        W_5 - W_4 &= \left(\sum_{\beta(x)-1 > a_3}\frac{1}{2^{\beta(x)+a_3+2}}-\sum_{\beta(x)-1>a_2>a_3} \frac{1}{2^{a_2+a_3+3}}\right) \left(\sum_{a_1 > \beta(x)-1} \frac{\wal_{2^{a_1}}(x)}{2^{a_1+1}} \right) \\
        &= \left(\sum_{\beta(x)-1 > a_3}\frac{1}{2^{\beta(x)+a_3+2}}-\sum_{\beta(x)-1>a_2>a_3} \frac{1}{2^{a_2+a_3+3}}\right)\left(s_1(x) - \sum_{\beta(x)-1 > a_1} \frac{1}{2^{a_1+1}} + \frac{1}{2^{\beta(x)}}\right)
    \end{align*}
    Similarity in the final two sums gives  
    \begin{align*}
        W_7 - W_6 &= \left(\frac{1}{2^{\beta(x)+1}}-\sum_{\beta(x)-1 > a_3} \frac{1}{2^{a_3+2}}\right) \left(\sum_{a_1>a_2 > \beta(x)-1} \frac{\wal_{2^{a_1}+2^{a_2}}(x)}{2^{a_1+a_2+2}}\right)
    \end{align*}
    where 
    \begin{align*}
        \sum_{a_1>a_2 > \beta(x)-1} \frac{\wal_{2^{a_1}+2^{a_2}}(x)}{2^{a_1+a_2+2}} &= \sum_{a_1>a_2} \frac{\wal_{2^{a_1}+2^{a_2}}(x)}{2^{a_1+a_2+2}} - \sum_{\beta(x)-1 > a_1 > a_2} \frac{1}{2^{a_1+a_2+2}} + \sum_{\beta(x)-1 > a_2} \frac{1}{2^{\beta(x)+a_2+1}} \\ &- \sum_{a_1 > \beta(x)-1 > a_2} \frac{\wal_{2^{a_1}}(x)}{2^{a_1+a_2+2}} + \sum_{a_1 > \beta(x)-1} \frac{\wal_{2^{a_1}}(x)}{2^{a_1+\beta(x)+1}} \\
        &= s_2(x) - \sum_{\beta(x)-1 > a_1 > a_2} \frac{1}{2^{a_1+a_2+2}} + \sum_{\beta(x)-1 > a_2} \frac{1}{2^{\beta(x)+a_2+1}}\\
        &+ \left(\frac{1}{2^{\beta(x)}} - \sum_{\beta(x)-1 > a_2} \frac{1}{2^{a_2+1}}\right)\left(s_1(x) - \sum_{\beta(x)-1 > a_1} \frac{1}{2^{a_1+1}} + \frac{1}{2^{\beta(x)}}\right).
    \end{align*}
    This implies 
    \begin{align*}
        \ts_4(x) &=  \frac{2}{3}\left(2-\beta(x)\right)x^3 + \left(1-5\ t_1(x)\right)x^2 - \frac{1}{9} \left(1 - 43 t_2(x)\right)x \\
        &-\frac{1}{48} \left(2-\beta(x)\right)\sum_{a_1} \frac{\wal_{2^{a_1}}(x)}{2^{3a_1}} -\frac{1}{294} \left(7\beta(x)+701 t_3(x)\right) +\frac{5}{98},
    \end{align*}
    from which the result follows.
\end{proof}



\end{document}